\documentclass[10pt]{article}
\usepackage[margin=1in]{geometry}
\usepackage[normalem]{ulem}
\usepackage{amsthm, amsmath, mathrsfs, amssymb, natbib, color, setspace, enumerate, titlesec, algorithmic, algorithm2e, 
prodint, 
subcaption, floatpag, tabularx, booktabs}
\usepackage{subcaption}
\usepackage{graphicx}

\usepackage[usenames,dvipsnames]{xcolor}
\usepackage[shortlabels]{enumitem}
\usepackage[toc,page]{appendix}
\titleformat*{\section}{\bfseries}
\titleformat*{\subsection}{\bfseries}
\titleformat*{\subsubsection}{\bfseries}
\titleformat*{\paragraph}{\bfseries}
\titleformat*{\subparagraph}{\bfseries}
\theoremstyle{plain}
\newtheorem{theorem}{Theorem}
\newtheorem{lemma}{Lemma}

\newtheorem{remark}{Remark}

\def\indicator{\mathbf{1}}
\makeatletter
\newcommand*{\textoverline}[1]{$\overline{\hbox{#1}}\m@th$}
\newcommand\cites[1]{\citeauthor{#1}'s\ (\citeyear{#1})}
\DeclareMathOperator*{\esssup}{ess\,sup}
\DeclareMathOperator*{\argmin}{arg\,min}

\newcommand\change[1]{\textcolor{black}{#1}}

\date{}
\begin{document}
\title{Simulation-based Value-at-Risk for Nonlinear Portfolios
\author{Junyao Chen, Tony Sit and Hoi Ying Wong\\~\\
\small{\textit{Department of Statistics, The Chinese University of Hong Kong}} ~\\
\scriptsize{yukichen@link.cuhk.edu.hk\quad tonysit@sta.cuhk.edu.hk \quad hywong@sta.cuhk.edu.hk}}}

\maketitle
\begin{abstract}
Value-at-risk (VaR) has been playing the role of a standard risk measure since its introduction. In practice, the delta-normal approach is usually adopted to approximate the VaR of portfolios with option positions. Its effectiveness, however, substantially diminishes when the portfolios concerned involve a high dimension of derivative positions with nonlinear payoffs; lack of closed form pricing solution for these potentially highly correlated, American-style derivatives further complicates the problem. This paper proposes a generic simulation-based algorithm for VaR estimation that can be easily applied to any existing procedures. Our proposal leverages cross-sectional information and applies variable selection techniques to simplify the existing simulation framework. 
  Asymptotic properties of the new approach demonstrate faster convergence due to the additional model selection component introduced. 
  We have also performed sets of numerical results that verify the effectiveness of our approach in comparison with some existing strategies.~
\end{abstract}

\section{Introduction}
One of the everyday challenges that financial institutions faces is re-evaluation of values and/or risk levels of their portfolios that mature some time in the future, which can generally be expressed in the form of
\begin{equation}
   \label{eq:genericpayoff}
      U(t, X) = \sup_{\tau \in \mathcal{T}}\mathbf{E}^{\mathbb{Q}}\left\{f(X_\tau)| \mathscr{F}_t \right\},
\end{equation}
where  \textcolor{black}{$t~(t>0)$ denotes the time, $f$ is a deterministic payoff function evaluated at the underlying asset value $X_t$, $\mathbb{Q}$ denotes a risk-neutral probability measure with respect to $\mathbb{P}$ and $\mathcal{T}$ is a family of stopping times. The filtration up to time $t$ is denoted as $\mathscr{F}_t$}. More importantly, based on these valuations, financial institutions need to calculate regulatory capitals in order to fulfill the requirements specified in Basel II for the banking industry \cite{bis2013a} or Solvency II for the insurance industry. 
Computation of regulatory capitals are closely related to Value-at-Risk (VaR), a fundamental quantity upon which some other coherent risk measures, including the expected shortfall \cite{Artzner_etal-1999-MF} are developed. Readers may refer to \cite{Kou_etal-2013-OR}, \cite{Kou_Peng-2016-OR} among others for further discussion. The main focus of this paper is to propose a more effective method for estimating VaRs.

While high-dimensional portfolios, or derivatives with large number of underlying assets, are 
common, a substantial portion of securities traded are derivatives with nonlinear payoffs; this renders the first-order, or even second-order, approximations insufficient for risk estimation. Evaluations of \eqref{eq:genericpayoff} and their corresponding risk measures hence become a non-trivial task. 
Given the fact that analytic solutions of \eqref{eq:genericpayoff} are hard to obtained in most cases, simulation is generally the only feasible resort; see \cite{Chan_Wong-2015, Glasserman-2003, Hong_etal-2014-ACM} amongst others. Despite their simplicity, simulation-based procedures may not be feasible because of its heavy computation burden. Although there have been new solutions on improving the computational efficiency (see, for instance, \citealp{Gramacy_Ludkovski-2015-SIAMJFM}), extensions to high-dimensional settings are not entirely straight-forward. To evaluate a $t_1$-day VaR with a particular statistical model chosen, one may carry out nested simulations.

An optimal allocation of computational effort for each layer (\citealp{Broadie_etal-2011-MS}) or simply reduce the number of simulated trails can also be applied for a more computationally economical alternative. However, curtailment of trials in either layer may lead to potentially substantial estimation bias and instability as pointed out in \cite{Bauer_etal-2012-ASTIN}.

In view of the aforementioned difficulties, current market practice is to calculate VaRs via Greek approximations such as the delta-normal and delta-gamma approximations; see \cite{Jorion-2006}. Performance of these approaches can sometimes be disappointing. In particular, for portfolios with highly nonlinear payoffs, the first-order approximation is far from sufficient in order to produce acceptably small errors. Besides, since all these Greeks are time-varying, delta-normal and delta-gamma approximations are reasonable only for portfolios with short investment horizons -- this can be rather restrictive for insurance companies as  the solvency capital ratios (SCR) required involve the one-year VaR valuation. Computation burden also poses a big concern as it increases exponentially with the number of stochastic variables included. Aggregation of huge biases from evaluating the Greeks numerically can also be potentially substantial.

To tackle the above challenges, \cite{Bauer_etal-2012-ASTIN} novelly proposed the use of the Least-squares Monte Carlo (LSM) approach to VaR computation based on  \cites{Longstaff_Schwartz-2001-RFS} seminal development for pricing American options. This approach, however, suffers from the curse of high-dimensionality when the number of underlying assets considered grows. The vast number of regressors generates highly volatile or even inconsistent coefficient estimates, which in turns leads to poor VaR estimates.

This paper incorporates the shrinkage idea in least-squares simulation for high-dimensional nonlinear portfolio VaRs. We shall demonstrate our proposal via  least absolute shrinkage and selection operator (LASSO; \citealp{Tibshirani-1996-JRSSB}), or equivalently the constrained $\ell_1$ minimization. Noteworthy, our proposal shares a similar view with \cite{Pun_Wong-2016-SIAMJFM}, \cite{Chiu_etal-2017-RA} and \cite{Pun_Wong-2019-EJOR} amongst others in the sense that the introduction of the LASSO penalty enables consistent estimation of the quantities of interest. For instance, \cite{Pun_Wong-2016-SIAMJFM} proved that the estimation errors of high-dimensional portfolio makes the optimal portfolio objective function diverge while our results demonstrate that, with appropriate shrinkage due to LASSO, the \cites{Longstaff_Schwartz-2001-RFS} approach can be properly implemented under high-dimensional cases.

\subsection{Summary of Contributions}
\textcolor{black}{In view of the popularity of the regression-based/Longstaff-Schwartz algorithm, our main goal is to study the corresponding convergence properties under the high dimensional setting. More specifically, this work contributes to the literature on  the following three aspects:}
\begin{enumerate}
   \item 
      \textcolor{black}{\textbf{Proper handling of issues due to high-dimensionality:} Amongst several works on analyzing the asymptotics of Longstaff-Schwartz algorithm, \cite{Clement_etal-2001-FS} provides theoretical justifications for regular cases with $p \ll N$, where $p$ and $N$ denote the dimension of the regressors and the sample size, respectively. One key assumption for the convergence results is that the model should include all the significant basis functions. Selection of basis functions is typically carried out rather subjectively and this assumption may not hold typically for assets with large numbers of underlying assets. To provide a more objective and systematic alternative, our approach leverages recent elegant results developed for variable selection so that we can consider a substantially larger number of covariates in the regression model without suffering issues due to high-dimensionality. Although various methods have been developed lately for high-dimensional linear regression such as the LASSO (see \citealp{Tibshirani-1996-JRSSB}), to the best of our knowledge, it is the first attempt to justify both theoretically and numerically how these variable selection tools can be incorporated in the Longstaff-Schwartz framework. The corresponding convergence results for various relevant estimates are also missing. To this end, we establish the relevant asymptotic results for both valuation and VaR estimation as the number of simulated paths $N$ goes to infinity together with the dimension in the regression model. Thus, for situations under which significant basis functions are not precisely known in advance, which are frequently encountered in various applications, the newly proposed shrinkage procedure, namely LASSO Least-squares Monte Carlo (LLSM), offers a higher chance of selecting influential basis functions in the regression than LSM. }
   \item 
      \textcolor{black}{\textbf{Theoretical construction: } We also enrich the proof by permitting estimation errors in the least-squares regression instead of assuming ideal estimates as required in \cite{Clement_etal-2001-FS}. This extension provides a more general discussion to the problem concerned. The framework developed lays down the foundation for other possible extensions, including the use of other variable selection methods besides LASSO as well as for other risk measures including expected shortfall (ES).} 
   \item
      \textcolor{black}{\textbf{Computational efficiency:} On the computation aspect, with the new variable selection element, the new proposal can handle an extensive number of basis functions based on asset prices and/or other risk factors and the LASSO component assists in selecting objectively and systematically the significant basis functions. LLSM significantly outperforms nested simulation and the Greek approximations in our numerical studies. The computational efficiency of LLSM is more prominent as the number of underlying stochastic variables increases. Numerical results show that it demands merely an additional $5\%$ (or 20\% including cross validation) of the total computation time to incorporate LASSO into the original LSM. The amount of additional computation time required declines as the dimension $p$ grows. The quality of resulting estimates is, however, dramatically improved; see Section \ref{st:numerical study}. }
\end{enumerate}

\subsection{Organization of the Paper}
The remainder of this paper is organized as follows. Section \ref{st:proof} elucidates the LLSM procedure, develops theoretical justifications for convergence results of LLSM and discusses further improvement of the new approach. Section \ref{st:numerical study} presents numerical studies on several derivatives with American features and nonlinear payoff functions. The performance of LLSM is demonstrated via a comprehensive comparison with existing methods and the oracle approach. Concluding remarks can be found in section \ref{st:conclusion}, followed by Appendix which presents the proofs for results discussed in Section \ref{st:proof}. Details of our numerical studies, including model specifications, are also included. 

\section{Methodology} \label{st:proof}
Our procedure \textcolor{black}{of LASSO Least-squares Monte Carlo (LLSM)} for a general portfolio with early exercise feature targets at $100(1-\alpha)\%$ $t_1$-day VaR over the investment horizon ranging from $T_0$ to $T$ during which stopping times denoted by $T_1,\ldots,T_L=T$ are covered. Noteworthy,  VaRs are not necessarily evaluated at stopping times, the procedure LLSM can handle a more generic $t_1$-day VaR with $t_1 \in (T_0, T_1)$. 

Similar to the celebrated \cite{Bauer_etal-2012-ASTIN} and \cite{Longstaff_Schwartz-2001-RFS} approaches, LLSM is formulated as a backward recursive procedure. In its first step, LLSM estimates the conditional expected option value via simulating paths. Based on these paths, regressions are carried out on the resulting option values. In contrast to the existing strategies, LLSM adds a variable selection step which allows an objective procedure for selecting the influential basis functions in the regression models considered. The corresponding regression result provides an approximation for the continuation value which can be compared to the early exercise value. Option values at different stopping times of all paths can then be evaluated, so can be the portfolio value as well as its VaR. Details of the algorithm for LLSM is summarized in Algorithm \ref{alg:general LLSM}.

\begin{algorithm}
  \caption{{LLSM for General Portfolios}
    \label{alg:general LLSM}}
  \hrule
  \begin{algorithmic}[1]
    \STATE {Identify the possible risk factors of the portfolio and denote them as a vector $X_t$, where the subscript $t$ denotes the time point at which the covariates recorded.}
    \STATE {Simulate $N$ sample paths of underlying stochastic variables $X_t$ for $t \in [T_0, t_1]$ under the physical measure, $\mathbb{P}$. For the remaining investment horizon $t_1$ to $T_L$, continue to simulate these paths from under the risk neutral measure(s) $\mathbb{Q}$. Realizations of $X_t$ at $t_1,T_1,\ldots,T_L$ are denoted as $X_{t_1},X_1,\ldots,X_L$ respectively.}
    \STATE {Initialize $\tau=L$ as the optimal stopping time indicator.}
    \FOR{$j \gets L-1 \textrm{ to } 1$}{
      \STATE Compute discounted continuation value at time $T_j$ by $C(T_j)=D(T_j,T_{\tau})A(T_{\tau})$ for each path, where $T_\tau$ is the optimal stopping time after $T_j$ that maximizes the portfolio value, $D(T_j,T_\tau)$ is the discount factor for the time period $(T_j,T_\tau)$, $A(T_\tau)$ is the immediate exercise value at $T_\tau$.
      \STATE Regress $C(T_j)$ on $L(X_j)$, where $L(X_j)$ is a vector of basis functions on $X_j$ with LASSO. Approximate $C(T_j)$ by the fitted value of the regression, $\hat{C}(T_j)$.
      \IF{$A(T_j)\geq \hat{C} (T_j)$}
        \STATE update $\tau=j$ for the corresponding path.
      \ENDIF
    }\ENDFOR
    \STATE  {Compute the portfolio value at $T_1$, denoted by $U_1$, by $U_1=D(T_1,T_\tau)A(T_\tau)$. }
    \STATE  {Regress $D(t_1,T_1)U_1$ on $L(X_{t_1})$ with LASSO. Approximate the portfolio value at $t_1$, denoted by $U_{t_1}$, by the fitted value of the regression $\hat{U}_{t_1}$.}
    \STATE  {Compute the loss $\ell=U_0-\hat{U}_{t_1}$. Rank $N$ realized losses and define the $\lceil\alpha N\rceil$th largest value as the estimate of $100(1-\alpha)\%$ $t_1$-day VaR.}
  \end{algorithmic}
  \hrule
\end{algorithm}

For the remainder of this section, we first introduce all notation needed for our subsequent discussion. As our VaR estimation procedure is developed upon prices evaluated from simulation, we first present the results of valuation in Section \ref{sst:valuation proof}, upon which VaR convergence can then be established; see Section \ref{sst:VaR proof}.

\subsection{Preliminaries and Notation}\label{sst:preliminaries}
Since the evaluation of $t_1$-day VaR depends on the estimate of portfolio value at $t_1$, which is derived from the portfolio values at stopping times $T_j$ for $j=1,\ldots,L$. To guarantee the convergence of VaR at $t_1$, we first develop the convergence results for product prices at stopping times $T_j$'s.

Assume an underlying complete probability space ($\Omega$,$\mathscr{F}$,$\mathbb{P}$) and finite time horizon ($0$,$T$), where $\Omega$ denotes the set of all possible realizations of the stochastic economy from time $0$ to $T$, $\mathscr{F}\triangleq\sigma(\Omega)=\mathscr{F}_T$ is the total information filtration accumulated up to $T=T_L$ with $T$ as the maximum maturity of all financial products in the portfolio. We discretize the time horizon into intervals $(T_{j-1},T_j)$ for $j=1,\ldots,L$ with equal length $\Delta t=T_j-T_{j-1}$ small enough so that potential exercise dates in the portfolio can be represented by some discrete time points $T_j$. Without loss of generality, we assume $T_j$ for $j=1,\ldots,L$ are the associated stopping times. Accordingly, we let $\mathscr{F}_j$ denote the information filtration up to time $T_j$. Denote $Z_j$ as the adapted payoff process of the portfolio and assume that $Z_j$ are square-integrable random variables for all $j$. At $T_j$, we let  $\{X_j \in  \mathbb{R}^{p_j}\mid X_j=\big (X_{j1},\ldots,X_{jp_j} \big )^\top \}$ be the $p_j$ underlying stochastic variables in the portfolio. As implied by our notation, the number of underlying stochastic variables at different $T_j$ is not necessarily fixed. One example is a portfolio which consists of interest rate products whose payoffs are functions of forward rates. For simplicity, we assume that $p_j \equiv p$ for $j = 1, \ldots, L$, and given $X_j$, there exists a deterministic payoff function $f$ such that $Z_j=f(T_j,X_j)$. The function $f$ can be nonlinear and/or discontinuous. Finally, we let $\mathcal {T}_{j,k}$ be the set of all possible stopping times $\{T_j,..,T_k\}$. Defined as the portfolio value at $T_j$, $U_j$ can be expressed in a form of conditional expectation as:
\begin{align} \label{eq:payoff conditional expectation}
U_j:=\sup_{\tau \in \mathcal{T}_{j,T}}\mathbf{E}^{\mathbb{Q}}\left\{f(T_\tau,X_\tau)\mid \mathscr{F}_j\right\},
\end{align}
where $\mathbb{Q}$ is a risk-neutral measure. In the sequel, the notation $\mathbb{Q}$ will be suppressed for the sake of simplicity. To illustrate the idea more effectively, we assume that there is only one optimal stopping time to be identified. If there is more than one derivative in the portfolio with different optimal stopping times, we may perform similar analysis by separating the portfolio into a linear combination of several elements, each of which has only one optimal stopping time that needs to be studied.

The formulation of the portfolio value $U_j$ defined in \eqref{eq:payoff conditional expectation} considers a fairly general setup and covers a wide range of assets in the market. 
 Our goal is to obtain an accurate estimate of $100(1-\alpha)\%$ $t_1$-day VaR, where $\alpha \in (0,1)$ is typically set to be $0.01$ or $0.05$. Assume, without loss of generality, that $t_1 \in (T_0,T_1)$ and that $T_0$ is the current time point at which $U_0$ is observed constant. If $t_1=T_1$, then we refer the VaR as VaR at a possible stopping time or else we refer it as VaR at a non-stopping time in general. In practice, most of the VaR's considered belong to the latter type. 

The $100(1-\alpha)\%$ $t_1$-day VaR is based on the estimation of portfolio value at future time point $t_1$. If $t_1=T_1$, $U_{t_1}$ can be computed through (\ref{eq:payoff conditional expectation}); if $t_1 \in (T_0,T_1)$, $U_{t_1}$ is defined as
\begin{align} \label{eq:payoff conditional expectation at t1}
U_{t_1}:=\mathbf{E}(U_1 \mid \mathscr{F}_{t_1})=\mathbf{E}\left[\sup_{\tau \in \mathcal{T}_{1,T}}\mathbf{E}\left\{f(T_\tau,X_\tau) \mid \mathscr{F}_1\right\}\bigg| \mathscr{F}_{t_1} \right].
\end{align}

Following classical optimal stopping theory \cite{Neveu-1975}, we introduce the Snell envelope and rewrite (\ref{eq:payoff conditional expectation}) as
\textcolor{black}{$$U_j:=\esssup_{\tau \in \mathcal{T}_{j,L}}\mathbf{E}(Z_\tau\mid \mathscr{F}_j)\quad j=0,1,\ldots,L, $$}
or equivalently as
\begin{displaymath}
U_j:= \left\{
\begin{array}{ll}
  Z_T, & j=L \\
  \max \{Z_j,\mathbf{E}(U_{j+1}\mid \mathscr{F}_j)\}, &0\leq j\leq L-1.
\end{array}
\right.
\end{displaymath}
If we define $\tau_j$ is the optimal stopping time after $T_j$, then $\tau_j:=\min\{k\geq j \mid U_k=Z_k\}$ in which case we can rewrite $U_j=\mathbf{E}(Z_{\tau_j}\mid\mathscr{F}_j),~j=0,1,\ldots,L$.

A backward approach is adopted to determine the optimal stopping time for each path. The rule can be stated by defining the dynamics of $\tau_j$ as,
\begin{displaymath}
\left\{
\begin{array}{ll}
\tau_T=T  \\
\tau_j=j \indicator_{\{Z_j\geq \mathbf{E}(Z_{\tau_{j+1}}\mid\mathscr{F}_j)\}}+\tau_{j+1} \indicator_{\{Z_j< \mathbf{E}(Z_{\tau_{j+1}}\mid \mathscr{F}_j) \}}, ~~0\leq j\leq L-1,
\end{array}
\right.
\end{displaymath}
where $\indicator_{\{\cdot\}}$ denotes the indicator function. Assume there is an $\mathscr{F}_j$-Markov chain $\{X_j\}$, $j=1,\ldots,L$, such that $Z_j=f(j,X_j)$ for some Borel functions $f(j,\cdot)$; then we have $U_j=g(T_j,X_j)$ for some function $g(j,\cdot)$ and $\mathbf{E}(Z_{\tau_{j+1}}\mid\mathscr{F}_j)=\mathbf{E}(Z_{\tau_{j+1}}\mid X_j)$ for $j=0,1,\ldots,L$. Note that in practice, $X_0$ and $U_0$ are both deterministic.

Denote $\{L_m(X_j)\}_{m\geq 1}$ as a sequence of measurable real-valued functions that serves the basis functions in the regression models. To numerically evaluate $\{\mathbf{E}(Z_{\tau_j})\}$, $j=1, 2, \ldots, L$ through a Monte Carlo procedure, we can simulate $N$ independent paths of the underlying risk factors of the Markov chain $\{X_j\}$. We define $X_j^{[i]}=(X_{j1}^{[i]},\ldots,X_{jp}^{[i]})^\top$ as the independent realizations of underlying stochastic variables at time $j$ for the $i$-th simulated path and $Z_j^{[i]}$ as the associated payoff for $j=1,2,\ldots,L$; $i=1,2,\ldots,N$ with $Z_j^{[i]}=f(T_j,X_j^{[i]})$.

In an attempt to approximate the conditional expectation $\mathbf{E}(Z_{\tau_{j+1}}\mid X_j)$ via a finite number of basis functions of $X_j$, we impose the following two conditions that appear in \cite{Clement_etal-2001-FS}:

\begin{enumerate}
   \item [(A1)] For $j=t_1,1,\ldots,L-1$, the sequence $\{L_m(X_j)\}_{m\geq 1}$ is total in $\mathcal{L}^2\{\sigma(X_j)\}$, where $\mathcal{L}^2\{\sigma(X_j)\}$ denotes the $\mathcal{L}_2$-space spanned by $\sigma(X_j)$.
   \item [(A2)]  For $j=t_1,1,\ldots,L-1$, if $\sum_{m=1}^M a_mL_m(X_j)=0$ a.s., then $a_m=0$ for $m=1,\ldots,M$, where $M$ denotes the number of basis functions included in the model.
\end{enumerate}

Under these two conditions, we can obtain coefficients vector $a_j^{[M]}$ such that $$\mathbf{E}(Z_{\tau_{j+1}}\mid \mathscr{F}_j)=\mathbf{E}(Z_{\tau_{j+1}}\mid X_j)=\lim_{M\to \infty} a_j^{[M]}\cdot L^{[M]}(X_j),$$ where $L^{[M]}(X_j)=(L_1(X_j),\ldots.,L_M(X_j))^\top$. To estimate the coefficients $a_j^{[M]}$, we assume
\begin{align} \label{eq:regression}
Z_{\tau_{j+1}}=a_j^{[M]}\cdot L^{[M]}(X_j)+\epsilon_j, ~~ j=1,\ldots,L-1,
\end{align}
where $\varepsilon_j$ is the error term. $a_j^{[M]}$ is known as the true coefficients in the regression. In line with the classical regression analysis, the gram matrix is defined as
\begin{align}\label{eq:gram matrix}
A_j^{[M,N]}=N^{-1}\sum_{i=1}^N\{L^{[M]}(X_j^{[i]})\} \{L^{[M]}(X_j^{[i]}) \}^\top.
\end{align}
We also define stopping times $\tau_j^{[M]}$ estimated by $M$ basis functions as
\begin{displaymath}
\left\{
\begin{array}{ll}
\tau_T^{[M]}=T  \\
\tau_j^{[M]}=j \indicator_{\{Z_j\geq a_j^{[M]}\cdot L^{[M]}(X_j)\}}+\tau_{j+1}^{[M]} \indicator_{\{Z_j< a_j^{[M]}\cdot L^{[M]}(X_j) \}},~~  0\leq j\leq L-1.
\end{array}
\right.
\end{displaymath}
Likewise, $\tau_j^{[i,M]} (j=1,\ldots,L)$ is used to denote the estimated stopping time with true coefficients in the regression for the $i$-th path. The estimated stopping time with LASSO estimated coefficients $a_j^{[M,N]}$ for the $i$-th path is denoted by $\tau_j^{[i,M,N]}$, where $a_j^{[M,N]}$ is defined as
$$a_j^{[M,N]}:=\argmin_{\alpha \in \rm I\!R^M}\left\{\| Z_{\tau_{j+1}^{[M,N]}}-\alpha\cdot L^{[M]}(X_j) \|_2^2+\lambda\|\alpha\|_{1}\right\}, ~~ j=1,2,\ldots,L-1,$$
with the penalty $\lambda$ depends on $M$ and $N$. In the sequel, we suppress the notation $\lambda^{[M,N]}$ for clearer presentation. Determining the optimal value for the regularization parameter is vital in terms of ensuring that the model performs well; typically, it is chosen by cross-validation. Our numerical procedure also adopts this approach for selecting a reasonable penalty.

To distinguish LASSO estimators from ordinary least-squares (OLS) estimators, we asterisk the associated symbols for all the parameters related to LSM. Accordingly, we have
$$a_j^{*[M,N]}:=\argmin_{\alpha \in \rm I\!R^M}\left\{\| Z_{\tau_{j+1}^{[M,N]}}-\alpha\cdot L^{[M]}(X_j) \|_2^2\right\}, ~~ j=1,2,\ldots,L-1$$
for the LSM approach. Based on the definition of estimated stopping times, we can define the portfolio value in (\ref{eq:payoff conditional expectation}) explained by $M$ basis functions with true coefficients as
\begin{equation*}
U_j^{[M]}:=\begin{cases} Z_T, & j = L,\\ Z_j \indicator_{\{Z_j\geq a_j^{[M]}\cdot L^{[M]}(X_j)\}}+U_{j+1}^{[M]} \indicator_{\{Z_j< a_j^{[M]}\cdot L^{[M]}(X_j) \}},& j=1,2,\ldots,L-1.\end{cases}
\end{equation*}

If we substitute $a_j^{[M,N]}$ into $a_j^{[M]}$ in the definition of $U_j^{[M]}$, we can obtain $U_j^{[M,N]}$, which is the portfolio value estimated by LLSM with $M$ basis functions and $N$ sample paths.

The following two subsections present the main contribution of this paper. Our first step is to establish the convergence result for valuation in Section \ref{sst:valuation proof}. Upon these consistent estimates of the derivative prices, the corresponding rates of convergence of VaR estimates are discussed in Section \ref{sst:VaR proof}. Despite the fact that techniques of handling high-dimensional data have been actively studied for the past two decades, to the best of our knowledge, there has not yet been any similar development in pricing/risk measure literature. All the new theorems presented subsequently compare the convergence rates for the traditional LSM and our proposal LLSM. The benefits of incorporating LASSO in the framework lies on the size of $M$, the number of basis functions, that can be handled by the model. Traditional methods like LSM performance can be significantly hindered when the dimension of the covariates grows, which in turns leads to non-invertibility of the associated gram matrix. Selection of basis functions are also conducted in a rather subjective manner. Our main result, Theorem \ref{thm:VaR convergency rate}, points out that when the number of sample paths is not significantly larger than the number of basis functions considered, the LSM approach can be outperformed by the new proposal.

\subsection{Convergence Results for Valuation} \label{sst:valuation proof}
To prove the convergence of a VaR estimate, we first establish the convergence result for valuation. The ultimate goal of valuation convergence is to prove
\begin{align} \label{eq:ultimate valuation convergence}
\mathbf{E}(Z_{\tau_j^{[M,N]}}|\mathscr{F}_j) \to \mathbf{E}(Z_{\tau_j}|\mathscr{F}_j) ~~ \text{ as } M,N \to \infty.
\end{align}

Similar to the treatment adopted in \cite{Clement_etal-2001-FS}, the convergence \eqref{eq:ultimate valuation convergence} can be established based on the two results of $\lim_{M \to \infty}U_j^{[M]}=U_j$ and $\lim_{N \to \infty}U_j^{[M,N]}=U_j^{[M]}$ for any fixed $M$. In particular, assume Condition (A1) is satisfied, for  $j = 1, 2, \ldots, L$, \cite{Clement_etal-2001-FS} shows that
\begin{equation}
\lim_{M \to \infty}\mathbf{E}(Z_{\tau_j^{[M]}}\mid \mathscr{F}_j)=\mathbf{E}(Z_{\tau_j}\mid \mathscr{F}_j).\label{thm:valuation convergence layer1}
\end{equation}

This result ensures the payoff $U_j^{[M]}$ estimated by regression on $M$ basis functions will converge to the true payoff $U_j$ as the number of basis functions $M$ tends to infinity. It is a consequence due to the total property of $L^2\{\sigma(X_j)\}$.

The next theorem stipulates that, under the same conditions that ensure valuation convergence of LSM, LLSM can achieve same rate of convergence for valuation at $T_j$ for $j=1,\ldots,L-1$. In other words, if the singularity problem can be solved through increasing $N$, the introduction of LASSO will not slow down the rate of convergence. Meanwhile, it suggests that under a weaker constraint on the singularity of the gram matrix, the almost sure convergence still holds for $U_j^{[M,N]}$. To examine the convergence of $U_j^{[M,N]}$ to $U_j^{[M]}$, three additional conditions are required:

\begin{enumerate}
   \item
      [(A3)] For $j=1,2,\ldots,L-1$, $i=1,2,\ldots,N$, realizations of $\epsilon_j$ in (\ref{eq:regression}) are i.i.d. with zero mean and finite variance.
   \item
      [(A4)] For $j=1,2,\ldots,L-1$, there exists a non-singular $M\times M$ matrix $C_j$ such that the gram matrix $A_j^{[M,N]}$ defined in (\ref{eq:gram matrix}) converges to $C_j$ as $N \to \infty$.
   \item
     [(A5)] (Compatibility Condition) Define the active set $S_0=\{m;a_{jm}^{[M]}\neq 0, m=1,2,\ldots,M\}$. The compatibility condition is met for the set $S_0$, if for some $\phi_0>0$ and for all $a^{[M]}$ satisfying $\|a^{[M]}_{S_0^c}\|_1\leq 3\|a^{[M]}_{S_0}\|_1$, it holds that
$\|a^{[M]}_{S_0}\|_1^2\leq \{a^{[M]}\}^\top A_j^{[M,N]} \{a^{[M]}\}\frac{s_0}{\phi_0^2}$, where $s_0=$card$(S_0)=|S_0|$.~\\
\end{enumerate}

\begin{theorem} \label{thm: valuation convergency}
Assume for $j=1,2,\ldots,L-1$, $\Pr \{a_j\cdot L^{[M]}(X_j)=Z_j\}=0$ and that Conditions (A1), (A2) and (A3) are satisfied. The LASSO estimators $a_j^{[M,N]}$ are obtained under the penalty with $\lambda = \mathcal{O}(\log M/N)$ and $\lambda/N=o(1)$.

\begin{enumerate}[(i)]
	\item
	   If Condition (A4) holds, then $U_j^{[M,N]}$ converges to $U_j^{[M]}$ almost surely.
    \item
       If Condition (A5) holds for the active set, then $U_j^{[M,N]}$ converges to $U_j^{[M]}$ almost surely also.
\end{enumerate}	
\end{theorem}
\begin{proof}
\textcolor{black}{Details of the proof can be found in \hyperref[sect:A1]{Appendix A.1}.}
\end{proof}

\begin{remark}
The assumption $\Pr \{a_j\cdot L^{[M]}_j(X_j)=Z_j\}=0$ is also required in \cite{Clement_etal-2001-FS}. To see the difference between LSM and LLSM, we observe that Theorem \ref{thm: valuation convergency} (\text{i}) also holds for $U_j^{*[M,N]}$ in LSM, but Theorem \ref{thm: valuation convergency} (\text{ii}) does not because without proper regularization, the associated gram matrix of the regression model in LSM will become singular.
\end{remark}
\begin{remark}
A similar version of Condition (A3) is also imposed in 
\cite{Clement_etal-2001-FS}. The definition of $a_j^{*[M,N]}$ in (2.11) of \cite{Clement_etal-2001-FS} assumes the gram matrix is invertible by default. If we adopt a more general definition of $a_j^{*[M,N]}$ that allows estimation error and takes the singularity problem into account, Condition (A4) is necessary for LSM. This condition is, however, rather restrictive since it requires the invertibility the gram matrix. The almost sure convergence property can still be maintained for the LLSM estimates even if we replace Condition (A4) with a less stringent constraint on the eigenvalues of the gram matrix. The Compatibility Condition (A5) (see also (6.4) of \cite{Buhlmann_vandeGeer-2011}) is similar to a constraint on the smallest eigenvalue of the gram matrix. This standard LASSO condition is a weaker condition which can be implied by Condition (A4). 
More discussion of the Compatibility Condition can also be found in \cite{Bickel_etal-2009-AoS, Koltchinskii-2009-AIHP} and \cite{Koltchinskii-2009-Bernoulli} amongst others.
\end{remark}

In Theorem \ref{thm: valuation convergency}, the additional LASSO component allows a substantially larger number of basis functions to be included in the model without corrupting the convergence of the estimated coefficient in the active set; see \cite{Buhlmann_vandeGeer-2011, Zhao_Yu-2006-JMLR}. We shall also see in Theorem \ref{thm:VaR convergency rate} the magnitude of $M$ that ensures convergence under this LASSO framework. Furthermore, the variable selection step in our model reduces the coefficient instability due to multicollinearity.

By \eqref{thm:valuation convergence layer1} and Theorem \ref{thm: valuation convergency}, we can see that the ultimate valuation convergence goal (\ref{eq:ultimate valuation convergence}) can be achieved almost surely in the following sense:
\begin{equation*}
\lim_{\substack{M \to \infty \\ N \to \infty}}\mathbf{E}(Z_{\tau_j^{[M,N]}}\mid \mathscr{F}_j) = \lim_{M \to \infty} \lim_{N \to \infty} \mathbf{E}(Z_{\tau_j^{[M,N]}}\mid \mathscr{F}_j) 
   =\lim_{M \to \infty} \mathbf{E}(Z_{\tau_j^{[M]}}\mid \mathscr{F}_j)
 = \mathbf{E}(Z_{\tau_j}\mid\mathscr{F}_j).
\end{equation*}

One may notice that the above induction may not be as straightforward as it appears because the value of $M$ is restricted by the choice of $N$. In fact, (\ref{eq:ultimate valuation convergence}) remains valid for some sufficiently large, yet finite, $M$, given that the $L^2\{\sigma(X_j)\}$ space is spanned by a finite number of basis functions. When the space $L^2\{\sigma(X_j)\}$ is spanned by a finite number of basis functions $L^{[M]}(X_j)$, the approach that can correctly choose all the unknown basis functions spanning $L^2\{\sigma(X_j)\}$ is desirable. If some of the necessary basis functions are excluded, convergence will never be obtained even when $N$ tends to infinity; on the other hand, if unnecessary basis functions are included, the increase in the number of coefficient parameters in the model may be poor due to numerically instability, eventually resulting in erroneous VaR estimates. The following theorem guarantees that LLSM can include more basis functions in the regression model than LSM for the same rate of convergence of the asset value.\\

\begin{theorem}\label{thm:valuation advantage}
Suppose the conditions in Theorem \ref{thm: valuation convergency} are satisfied and the Irrepresentable Condition in the sense of \cite{Zhao_Yu-2006-JMLR} holds for the active sets, $|S_0|=s_0<\infty$ for $j=1,\ldots,L-1$; see also Appendix for the definition of Irrepresentable Condition.  If a finite set of $M_1$ basis functions are initially included in the regression with $M_1$ sufficiently large so that $S_0 \subseteq S_0^{[M_1]}$, then there exists $M\leq M_1<\infty$ such that,
$$U_j^{*[M_1,N]}  \overset{as} \to U_j~~\text{ and } ~~ U_j^{[M,N]}  \overset{as} \to U_j~~ \text{ as } N \to \infty.$$
\end{theorem}
\begin{proof}
\textcolor{black}{Details of the proof can be found in  \hyperref[sect:A2]{Appendix A.2}.}
\end{proof}

Theorem \ref{thm:valuation advantage} ensures that, given a suitable penalty $\lambda$, one can carry out the valuation procedure with finite number of basis functions and obtain the same convergence result as $N$ increases. Furthermore, the number of basis functions considered in LLSM never exceeds that considered in LSM for the same convergence result based on the same initial set of basis functions. The Irrepresentable Condition is a stronger condition that implies the compatibility Condition. It depends on the gram matrix and the signs of true coefficients; see \cite{Buhlmann_vandeGeer-2011} for more discussion.

The above result also concludes that the number of basis functions needed to obtain convergence in LLSM is upper bounded by that required by LSM. Fewer basis functions in the regression model implies that there will be less estimation error given the same computation budget. Admittedly, there is no guarantee that one can include all the influential basis functions that span $L^2\{\sigma(X_j)\}$ in the regression model. Nonetheless, given the same computation budget $N$, LLSM enables users to initially include and screen more basis functions; \change{see also Theorem \ref{thm:VaR convergency rate}}.\\

\subsection{Convergence Results for VaR}\label{sst:VaR proof}
Given the valuation convergence results presented in Section \ref{sst:valuation proof}, we now establish the corresponding convergence properties of the VaR estimate proposed. As discussed earlier, the properties of a $t$-day VaR with $t$ as a stopping time are different from cases where $t$ is not a stopping time. In this section, we present Theorem \ref{thm: VaR convergency} which ensures the convergence of VaR at possible stopping times. The specific rates of convergence of VaRs at non-stopping times evaluated via LSM and LLSM are derived in Theorems \ref{lemma:VaR convergence rate} and \ref{thm:VaR convergency rate} respectively. \\

\begin{theorem}\label{thm: VaR convergency}
For $j=1,\ldots,L-1$, if conditions in Theorem \ref{thm: valuation convergency} (\text{i}) are satisfied, then
$$\text{VaR}_j^{[M,N]} \to \text{VaR}_j^{[M]}~~ \text{ as } N \to \infty,$$
where $\text{VaR}_j^{[M,N]}$ and $\text{VaR}_j^{[M]}$ are defined as,
\begin{eqnarray*}
\text{VaR}_j^{[M,N]} & \triangleq & \inf_{x\in \rm I\!R}\left\{\Pr(U_0-U_j^{[M,N]}<-x)<\alpha \right\},\\
\text{VaR}_j^{[M]} & \triangleq & \inf_{x\in \rm I\!R}\left\{\Pr(U_0-U_j^{[M]}<-x)<\alpha \right\}.
\end{eqnarray*}
\end{theorem}
\begin{proof}
\textcolor{black}{Details of the proof can be found in  \hyperref[sect:A3]{Appendix A.3}.}
\end{proof}

\begin{remark}
This theorem also holds for $\text{VaR}_j^{*[M,N]}$ derived from LSM. A similar convergence result still holds for $\text{VaR}_j^{[M,N]}$ if we substitute the Compatibility Condition, a weaker condition, for Condition A4. It is, however, not true for $\text{VaR}_j^{*[M,N]}$.~
\end{remark}

Theorem \ref{thm: VaR convergency} proves the convergence of VaR estimates by LLSM at stopping times. Both $\text{VaR}_j^{[M,N]}$ and $\text{VaR}_j^{*[M,N]}$ converge at the rate of $\mathcal{O}(N^{-1})$; \textit{c.f.} Proposition 3.2 of 
\cite{Bauer_etal-2012-ASTIN}. However, in most cases, we need the convergence result for $t_1$-day VaR with a non-stopping time $t_1$. In a typical setting, for instance, a risk manager has to compute a $10$-day VaR in order to fulfill the Basel \text{II} regulations. In this case, $t_1=10$-day and $t_1\notin \mathcal{T}_{0,T}$; the convergence of $\text{VaR}^{[M,N]}_{t_1}$ to the $\text{VaR}^{[M]}_{t_1}$ is obviously important. To achieve this, we provide Theorems \ref{lemma:VaR convergence rate} and \ref{thm:VaR convergency rate} which guarantee that, under some mild conditions, VaR estimates by LLSM at non-stopping times converge at a faster rate than the counterparts obtained by LSM. This theorem explains why LLSM always outperforms LSM when we compute $95\%$ $10$-day VaR in our numerical studies. 

To handle calculations related to non-stopping time, we write the estimate of $Z_{\tau_1}$ as a combination of basis functions, \textit{viz.} $$Z_{\tau_1^{[M]}}=a_{t_1}^{[M]}\cdot L^{[M]}(X_{t_1})+\epsilon_{t_1},$$ where $a_{t_1}^{[M]}$ is referred to the true coefficients in the regression at $t_1$ and $\epsilon_{t_1}$ denotes the error term with zero mean and finite variance. Note that $Z_{\tau_1^{[M]}}$ serves as the response in the regression, indicating that true coefficients are used in each regression to estimate $\tau_1^{[M]}$. The LASSO estimates are defined correspondingly as $a_{t_1}^{[M,N]} = \argmin_{\alpha \in \rm I\!R^M} \left\{\|Z_{\tau_1^{[M,N]}}-\alpha \cdot L^{[M]}(X_{t_1})\|_2^2+\lambda\|\alpha\|_1\right\}$, where $Z_{\tau_1^{[M,N]}}$ is the response in the regression. The true coefficients in the same regression is defined as $\tilde{a}_{t_1}^{[M,N]}$. The corresponding OLS estimates, namely $a_{t_1}^{*[M,N]}$ and $\tilde{a}_{t_1}^{*[M,N]}$, can be obtained by substituting $Z_{\tau_1^{*[M,N]}}$ with $Z_{\tau_1^{[M,N]}}$ as the response in the regression.

The pricing error at $t_1$ is composed of two components. One is the estimation error that comes from the regression at $t_1$, denoted by $\bigg|(a_{t_1}^{[M,N]}-\tilde{a}_{t_1}^{[M,N]})\cdot L^{[M]}(X_{t_1})\bigg|$; the other is the estimation error of $Z_{\tau_1^{[M]}}$, denoted by $\bigg|N^{-1}\sum_{i=1}^N(Z_{\tau_1^{[i,M,N]}}^{[i]}-Z_{\tau_1^{[i,M]}}^{[i]})\bigg|$ with the superscript $i$ in this notation indicates the $i$th realization of the corresponding random variables. Although both $a_{t_1}^{[M]}$ and $\tilde{a}_{t_1}^{[M,N]}$ are called true coefficients, different responses are used as dependent variables in the corresponding regression. Due to the fact that the definition of $U_{t_1}^{[M]}$ is different from that of $U_j^{[M]}$, $j=0,\ldots,L$, we cannot trivially apply Theorem \ref{thm: VaR convergency} to the proof of VaR convergence at $t_1$.

To tackle this problem, we define
\begin{align*}
\bar{W}\triangleq& N^{-1}\sum_{i=1}^N\left \{ a_{t_1}^{[M]}\cdot L^{[M]}(X_{t_1}^{[i]})-a_{t_1}^{[M,N]}\cdot L^{[M]}(X_{t_1}^{[i]}) \right\}\\
\bar{W}^*\triangleq& N^{-1}\sum_{i=1}^N\left \{a_{t_1}^{[M]}\cdot L^{[M]}(X_{t_1}^{[i]})-a_{t_1}^{*[M,N]}\cdot L^{[M]}(X_{t_1}^{[i]}) \right\}
\end{align*}
as the average pricing error for LASSO and OLS, respectively. We also define $W=\sqrt{N}\bar{W}$ and $W^*=\sqrt{N}\bar{W}^*$. Let $g_N(\cdot,\cdot), g(\cdot)$ and $g_N(\cdot)$ denote the joint pdf of $U_{t_1}^{[M]}$ and $W$, the marginal pdf of $U_{t_1}^{[M]}$ and the pdf of $U_{t_1}^{[M,N]}$, respectively.
To ensure the VaR convergence for the nested simulation and for LSM, the following condition that imposes some restriction on the distribution of $W$ and $W^*$ is required; see \cite{Gordy_Juneja-2010-MS} and \cite{Bauer_etal-2012-ASTIN}.\\

We say that Condition (A6) holds for random variable $W$ if both of the following are satisfied:
\begin{enumerate}[i.]
   \item
      The joint pdf $g_N(\cdot,\cdot)$ of $U_{t_1}^{[M]}$ and $W$ and its partial derivatives $\frac{\partial}{\partial u}g_N(u,w)$, $\frac{\partial ^2}{\partial u^2}g_N(u,w)$ exist for each $N$ and for all sets of $(u,w)$.
   \item
      For $N\geq 1$, there exist non-negative functions $p_{0,N}(\cdot)$, $p_{1,N}(\cdot)$, $p_{2,N}(\cdot)$ such that for all $(u,w)$,
\begin{equation*}
g_N(u,w) \leq p_{0,N}(w), \quad \bigg|\frac{\partial}{\partial u}g_N(u,w)\bigg| \leq p_{1,N}(w), \quad \bigg|\frac{\partial ^2}{\partial u^2}g_N(u,w)\bigg| \leq p_{2,N}(w).
\end{equation*}
In addition,
$$\sup_N \int_{-\infty}^\infty |w|^rp_{i,N}(w)dw<\infty ~~ \text{ for }i=0,1,2 \text{~~and~~}0\leq r\leq 4.$$
\end{enumerate}

This condition generally holds for large portfolios where there are at least a few positions that have sufficiently smooth payoffs; see \cite{Gordy_Juneja-2010-MS}. To compare the performance of LLSM and LSM, we introduce Theorem \ref{thm:VaR convergency rate} that shows the convergence rate of $\text{VaR}_{t_1}^{[M,N]}$ and $\text{VaR}_{t_1}^{*[M,N]}$.

\begin{theorem}\label{thm:VaR convergency rate}
If conditions in Theorem \ref{thm: valuation convergency} (\text{i}) are satisfied, Condition (A6) holds for $W$ and $W^*$, $\text{VaR}_{t_1}^{[M,N]}$ by LLSM $\text{VaR}_{t_1}^{*[M,N]}$ by LSM will converge to $\text{VaR}_{t_1}^{[M]}$ in the following sense,
\begin{eqnarray*}
\text{VaR}_{t_1}^{[M,N]}-\text{VaR}_{t_1}^{[M]}& = &\mathcal{O}\left(\sqrt{\frac{s_0\log M}{N\phi_0^2}}\right)+\left[\frac{g(\tilde{v})g'(v)}{g(v)}-\frac{1}{g(v)}\frac{d}{dv}g(v) \right]\mathcal{O}\left(\frac{s_0\log M}{N\phi_0^2}\right)\\
&& +o\left(N^{-1}\right),  \\
\text{VaR}_{t_1}^{*[M,N]}-\text{VaR}_{t_1}^{[M]}&=&\frac{1}{g(v)}\frac{d}{dv}g(v)\mathcal{O}\left(\frac{M}{N}\right)+o\left(N^{-1}\right),
\end{eqnarray*}
where $v=\text{VaR}_{t_1}^{[M]}-U_0$ and $\tilde{v} \in [v-w/\sqrt{N},v]$. Furthermore, \newline if $N=o\left(\frac{M^2\phi_0^2}{s_0\log M}+2M+\frac{s_0\log M}{\phi_0^2}\right)$, we will have
$\frac{\text{VaR}_{t_1}^{[M,N]}-\text{VaR}_{t_1}^{[M]}}{\text{VaR}_{t_1}^{*[M,N]}-\text{VaR}_{t_1}^{[M]}}=o(1).$~\\
\end{theorem}
\begin{proof}
\textcolor{black}{Details of the proof can be found in  \hyperref[sect:A4]{Appendix A.4}.}
\end{proof}

\begin{remark}
Theorem \ref{thm:VaR convergency rate} still hold if we substitute the Compatibility Condition for Condition (A4). Note that in this case, $\text{VaR}_{t_1}^{[M,N]}$ will still converge whereas $\text{VaR}_{t_1}^{*[M,N]}$ will diverge. ~\\
\end{remark}

As we can see in this theorem, LLSM allows us to include $o(\exp(N))$ basis functions whereas LSM can only handle at most $o(N)$ for convergence. If the gram matrix is non-singular, LLSM yields a faster VaR convergence rate than LSM under restriction of $N=o\left(\frac{M^2\phi_0^2}{s_0\log M}+2M+\frac{s_0\log M}{\phi_0^2}\right)$. Such a growth rate of $N$ can be explained in the following two aspects. Firstly, this choice of $N$ means that the number of sample paths available cannot be infinitely large due to a given computation budget. Under a high-dimensional setting with $M$ large, $N$ can hardly be larger than $\mathcal{O}(M^2)$. Secondly, if we have enough resources so that $N>\mathcal{O}(M^2)$, the LASSO component may not be necessary given the non-singularity of the gram matrix and abundant sample paths. LASSO has been well-known for its application in high-dimensional statistics, but bias would arise if we impose a penalty in the minimization process in an unnecessary case when $N$ is sufficiently large and the gram matrix is non-singular.

\section{Numerical Studies} \label{st:numerical study}
Our quantity of interest is the $95\%$ $10$-day VaR for portfolios with nonlinear payoffs. Back testing is performed to evaluate the performance of different approaches when oracle benchmarks are available. In this section, the penalty used in LASSO is determined by 20-fold cross-validation to minimize the mean cross-validated error given a loss function. 
\textcolor{black}{We refer the nested simulation in \cite{Gordy_Juneja-2010-MS} as the estimated oracle approach. If, in the inner simulation, a closed form solution is available for evaluating the portfolio at $t_1=10$-day, we define the approach as the true oracle approach. The Greeks involved in the delta-normal approach are computed numerically via center finite difference method.} 

Although we consider VaR estimation of individual products, the idea of VaR evaluation can be extended from a single derivative to a high-dimensional portfolio by including additional risk factors as the underlying stochastic variables in the regression. Common risk factors are simulated once and only one regression will be performed at each possible stopping times and $t_1$ to evaluate the value of the whole portfolio. Specifically, to make the results more directly comparable with those presented in \cite{Longstaff_Schwartz-2001-RFS}, we adopted polynomials up the three order as our basis functions $L(X)$ for all examples. \textcolor{black}{In the following examples, we shall assume that the return series follow multivariate Gaussian distributions. They are constructed in this way such that we can easily benchmark our performance with existing procedures, especially those which rely on the closed-form solutions under such settings. Noteworthy, however, our formulation does not require joint normality assumption for the return series. Because of the non-parametric nature of our estimate, our proposal can be extended to non-elliptical world fairly easily because of the ranking step stated in Step 13 in Algorithm \ref{alg:general LLSM}.}

\subsection{Rainbow Option}\label{sst:rainbow option}
Rainbow options are one of the most commonly traded exotic options whose payoff functions depend on more than one underlying risky assets. In this section, we consider a variation of ``call on min" rainbow option with ten stocks as its underlying risky assets. The long side will receive a positive profit if the minimum ratio return of ten underlying stocks exceeds a predefined strike price. In other words, the payoff at maturity is expressed as
\begin{align*}
100\max\left(\min_i \frac{S_{iT}}{S_{i0}}-K,0\right),
\end{align*}
where $S_{i0}$ denotes the current price for the $i$th underlying stock. The constant $100$ in the payoff function is arbitrary for illustration to standardize the payoff at maturity. In order to derive a benchmark based on the closed form solution for pricing, we assume the underlying stock prices follow the \cite{Black_Scholes-1973-JPE} model.  The closed form solution is discussed in \cite{Johnson-1987-JFQA}. Corresponding details are provided in the Appendix; see Section B2.

The VaR estimates given by different approaches are summarized in Table \ref{table:VaR of rainbow option}. The strike is selected to ensure the rainbow option is at-the-money, a situation in which delta-normal approximation may face challenges due to non-differentiability at the price that corresponds to unit moneyness. We chose the maturity $T$ to be 270 days in this example. The number of sample paths generated in each approach is $N=10,000$ and the number of paths in the inner layer of the estimated oracle approach is $N_2=50,000$.

Since we can obtain one estimate of VaR in the estimated oracle approach, there is no observation of the standard deviation. Except for the oracle approaches, each methodology is repeated for 500 iterations in order to study the distribution of the VaR estimates. 
\textcolor{black}{Procedures labelled with $^\dagger$ adopt the closed form solution for all the pricing involved.}

The computation time indicates the time needed for an approach to obtain one VaR estimate yielded from a computer with Intel Core i5-5200U, CPU 2.2 GHz and RAM 8GB.

\begin{table}[]\footnotesize
\caption{10-day 95\% VaR of Rainbow Option}
\label{table:VaR of rainbow option}
\begin{tabular}{@{}lccccc@{}}
\hline
& Mean    & Median      & Standard Deviation   & Back Testing & Time (in seconds) \\ \hline
LSM                                       & 1.78698 & 1.78664 & 0.11482 & 0.0290       & 18.62             \\
LLSM& 1.61693 & 1.61108 & 0.10932 & 0.0442 & 20.25             \\
Delta-normal$^\dagger$                & 1.74591 & 1.74591 & -           & 0.0329       & 7.62             \\
\textcolor{black}{Delta-gamma$^\dagger$} & 5.00858 & 5.00858 & -           & 0       & 66.36             	\\
Oracle                 & 1.58089 & 1.58089     & -           & 0.0483       & 163,850 \\
Oracle$^\dagger$                     & 1.56778 & 1.56778     & -           & 0.0500         & 3,634     \\ \hline
\end{tabular}
\end{table}

As shown in Table \ref{table:VaR of rainbow option}, only a small amount of additional computation is required to carry out the variable selection, even though 20-fold cross validation is adopted for LLSM. Upon our VaR estimates, the back testing procedure was carried out by comparing the estimates with the unrealized P\&L's of the simulated prices evaluated based on the closed form formulas. Percentages of losses that exceed the VaR estimates are reported. According to Table \ref{table:VaR of rainbow option}, we can see that it is worthwhile to carry out the additional LASSO variable selection procedure since the back testing results are dramatically improved from $2.90\%$ to $4.42\%$. \textcolor{black}{For a fair comparison, both Delta-normal and Delta-gamma approaches apply the finite difference method for Greeks calculations. We observe biased estimates for Greeks with higher orders and significantly heavier computational burden as the number of Greeks increases. The back testing results of $3.29\%$ and $0\%$ in the Delta-normal and Delta-gamma approach can be improved to $5.23\%$ and $6.18\%$ respectively if the closed form solution is applied to Greeks computing. The Delta-gamma approach has poorer performance because of the biases accumulated in repeated numerical approximations of the differentials.} 

These results verify that even with a short horizon, \textcolor{black}{neither first nor second-order approximations is insufficient for estimating VaR's of derivatives with nonlinear payoffs.} The discrepancy is even more prominent when the derivatives are nearly at-the-money.

\subsection{European Swaption}\label{sst:european swaption}
Swaptions are among the most liquidly traded interest rate derivatives in the financial market. 

Consider a European payer $20$ NC (``non-call/lock-out'' period) $2$ swaption whose underlying swap has a final tenor of 20 years. We adopt the Lognormal Forward LIBOR Model (LFM) as the underlying model for the forward rates in the swaption. Same definitions and calibrations are adopted from \cite{Brigo_Mercurio-2007}. Denote $L(t,T)$ as the spot interest rate prevailing at time $t$ for the maturity $T$ and $P(t,T)$ as the zero-coupon bond price delta-normalat time $t$ with payment at maturity $T$. The forward rates are denoted by $L_i(t) \equiv L(t,T_{i-1},T_i)$, where $i=1,\ldots,20$. The forward rates dynamics in the LFM are defined in Proposition 6.3.1 in \cite{Brigo_Mercurio-2007}.

Given a notional amount of $N=1,000$ and the swap rate $K$, the payoff to the holder at $T_i$ is
\begin{align}  \label{eq:swaption payoff}
A(T_i)=1000 \cdot \mathbf{E}^{\mathbb{Q}^i}\left[\left\{\sum_{j=i+1}^{20}D(T_i,T_j)\delta_j(L_j(T_i)-K)\right\}^+~\middle|~\mathcal{F}_i \right],
\end{align}
where $i=2, \ldots,20$, $\delta_j=\delta(T_{j-1},T_j)$ is the discrete time interval, $D(T_i,T_j)$ is the discount factor for time period of $(T_i,T_j)$ and $\mathbb{Q}^i$ is a forward-adjusted measure corresponding to time $T_i$. More details about the model and parameters calibration can be found in the Appendix.

In the numerical study of swaption in \cite{Longstaff_Schwartz-2001-RFS}, the basis functions are subjectively selected to be a constant, the first three powers of the discounted price of the swaption at $t$, and the first power of all immatured zero coupon bond prices with final maturity dates up to and including $T_{20}$. We refer LSM with subjectively selected basis functions as SLSM. This method can potentially be unreliable as it performs a subjective apriori variable selection. For general products with a large number of underlying assets across different asset classes, the selection may not be as straight forward as the case for swaption.

We denote GLSM as LSM that specifically includes the first three orders of risk factors and second order of cross terms of these risk factors in the regression model. Note that GLSM does not include cross terms up to third order as in LSM. We allow this loose restriction on the order of basis functions to avoid that LSM fails to get OLS coefficient estimates due to over-parameterization.

The swap rate of the underlying swap is determined at $T_0$ to guarantee the swaption at-the-money. The numbers of sample paths in each approach are $N=5,000$. The number of paths in the outer layer and inner layer is $N_1=30,000$ and $N_2=30,000$ respectively. To ensure the estimated oracle approach offers a stable estimation, we have examined and selected different number of intensive simulation paths. We choose sufficient large $N_1$ and $N_2$ so that no significant change is observed with any further increment. Four approaches except the oracle approach are repeated $500$ times to get sample statistics. The computation time indicates the mean time needed for carrying out one round of iteration.

\begin{table}[]\footnotesize
\centering
\caption{10-day 95\% VaR of European Swaption}
\label{table:VaR of European swaption}
\begin{tabular}{@{}lccccc@{}}
\hline
  & Mean        & Median     & Standard Deviation         & Back Testing & Time (in seconds) \\ \hline
SLSM      & 8.94452   & 9.02583  & 1.45876 & 0.0200       &  13.07    \\
GLSM                & 19.8179  & 19.8070  & 1.23116 & 0.0000            & 17.25    \\
LLSM           & 7.02251   & 7.15271   & 1.62806  & 0.0505       & 25.16    \\
Delta-normal                   & 8.88734 & 8.88922 & 3.45005  & 0.0208       & 285.78    \\
Oracle                         & 7.04185  & 7.04185  & -          & 0.0500         & 242,200    \\ \hline
\end{tabular}

\end{table}

As shown in Table \ref{table:VaR of European swaption}, the computation time needed for the delta-normal approach is significantly longer than other approaches. \textcolor{black}{This is due to the fact that the best effort available to evaluate the portfolio value at $T_0$ is the estimated oracle approach. Nested simulation is required for each shift in each of the $18$ underlying risk factors at $T_0$ for the delta-normal approach}. The application of the estimated oracle approach is rather limited due to its computational burden: Even for a European swaption, it demands approximately three days to calculate one estimate of VaR.

The standard deviations for the first three methods are close but significantly larger than that obtained from the delta-normal approach. Despite the small standard deviation of the estimates given by the delta-normal approach, it incurs rather large biases which cast doubt on the accuracy of its performance. The boxplot shown in Figure \ref{fig:Eboxplot} summarizes the distribution of the VaR estimates obtained by the first four approaches. The dots in each approach represent VaR estimates in 500 experiments. The dash line draws the VaR obtained by the estimated oracle approach.

Among these five methods, GLSM performs worst. For the delta-normal approach, it produces estimates with a smaller bias, but with abnormally small variance. In the $500$ experiments, no results from the delta-normal approach or GLSM produces VaR estimate that is close to the oracle VaR. For SLSM, the dash line is located beyond the $25\%$ quantile of the distribution, indicating that this approach still has a small probability if getting an accurate VaR in one experiment. Regarding LLSM, the median of the distribution is closer to the dash line, indicating that the bias is small. Variance of this approach is also reasonable, in the sense that the dash line crosses the distribution within the range of 25\% and 75\% quantiles.

\begin{figure}[h!]
\centering
\includegraphics[width=0.6\textwidth]
{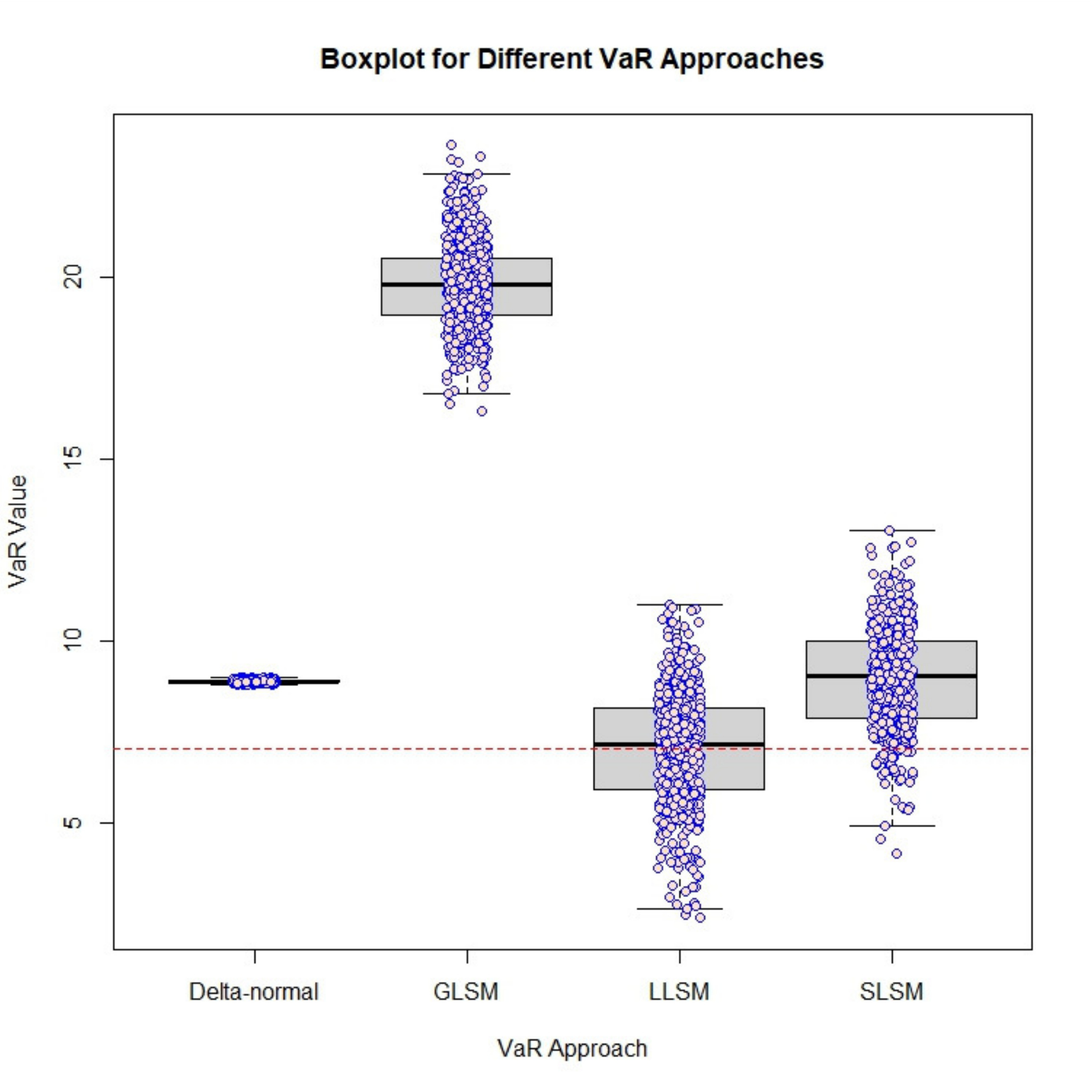}
\caption{VaR Boxplot for European Swaption}
\label{fig:Eboxplot}
\end{figure}

The performance can be evaluated through the back testing result summarized in Table \ref{table:VaR of European swaption}. Consistent with the analysis depicted in Figure \ref{fig:Eboxplot}, GLSM severely overestimates VaR, resulting a back testing result of $0$. The SLSM and the delta-normal approach have similar biases and similar back testing results of around $2\%$. Their back testing results are not satisfactory either because the estimated VaRs are too conservative, which consequently requires extra unnecessary capital reserves. LLSM, although underestimates VaR, performs much better with a back testing result of $5.05\%$. Overall, LLSM offers the best performance among the four approaches.

\subsection{Bermudan swaption} \label{sst:Bermudan swaption}

Since LLSM is applicable to portfolios with American features, we extend the previous example to Bermudan swaptions. Consider a Bermudan payer 20 NC 2 swaption. The payoff to the holder at $T_i$, $i=2,\ldots,20$ is defined as (\ref{eq:swaption payoff}). Each approach is repeated for $100$ times. Since it is not practical to perform nested simulation to derive oracle initial value, we applied SLSM with sufficiently large number of paths to determine the initial value of the swaption. Other settings are the same as in the previous study.

\begin{table}[]\footnotesize
\centering
\caption{10-day 95\% VaR of Bermudan Swaption}
\label{table:VaR of Bermudan swaption}
\begin{tabular}{@{}lcccc@{}}
\hline
& Mean       & Median     & SD & Time (in seconds) \\ \hline
SLSM         & 8.38623 & 8.48900 & 1.59813 & 195.00        \\
GLSM         & 21.0271 & 21.1145 & 1.51290 & 226.62       \\
LLSM         & 5.01065 & 4.98452 & 1.86820 & 270.01      \\
Delta-normal & 189.649 & 7.87119 & 371.980 &
8,807.50  \\ \hline
\end{tabular}
\end{table}

As shown in Table \ref{table:VaR of Bermudan swaption}, the computation time for the delta-normal approach is significantly larger than other approaches due to re-valuations required for each shift in the underlying risk factors. \textcolor{black}{SLSM is used in evaluating the portfolio value at $T_0$ in the delta-normal approach since it is the best effort available for swaptions with Bermudan feature}  In some iterations, some of the deltas are especially large, thus leads to inflated trails. As we can see in Figure \ref{fig:Bboxplot}, the VaR calculated from the delta-normal approach is heavily right-skewed with a large number of outliers, whereas the VaR from other four approaches appears to be symmetrically distributed with little outliers. The large standard deviation also indicates that the delta-normal approach lacks statistical efficiency.

\begin{figure}
   \begin{minipage}[b]{.5\linewidth}
     \centering
     \includegraphics[width=1\textwidth]{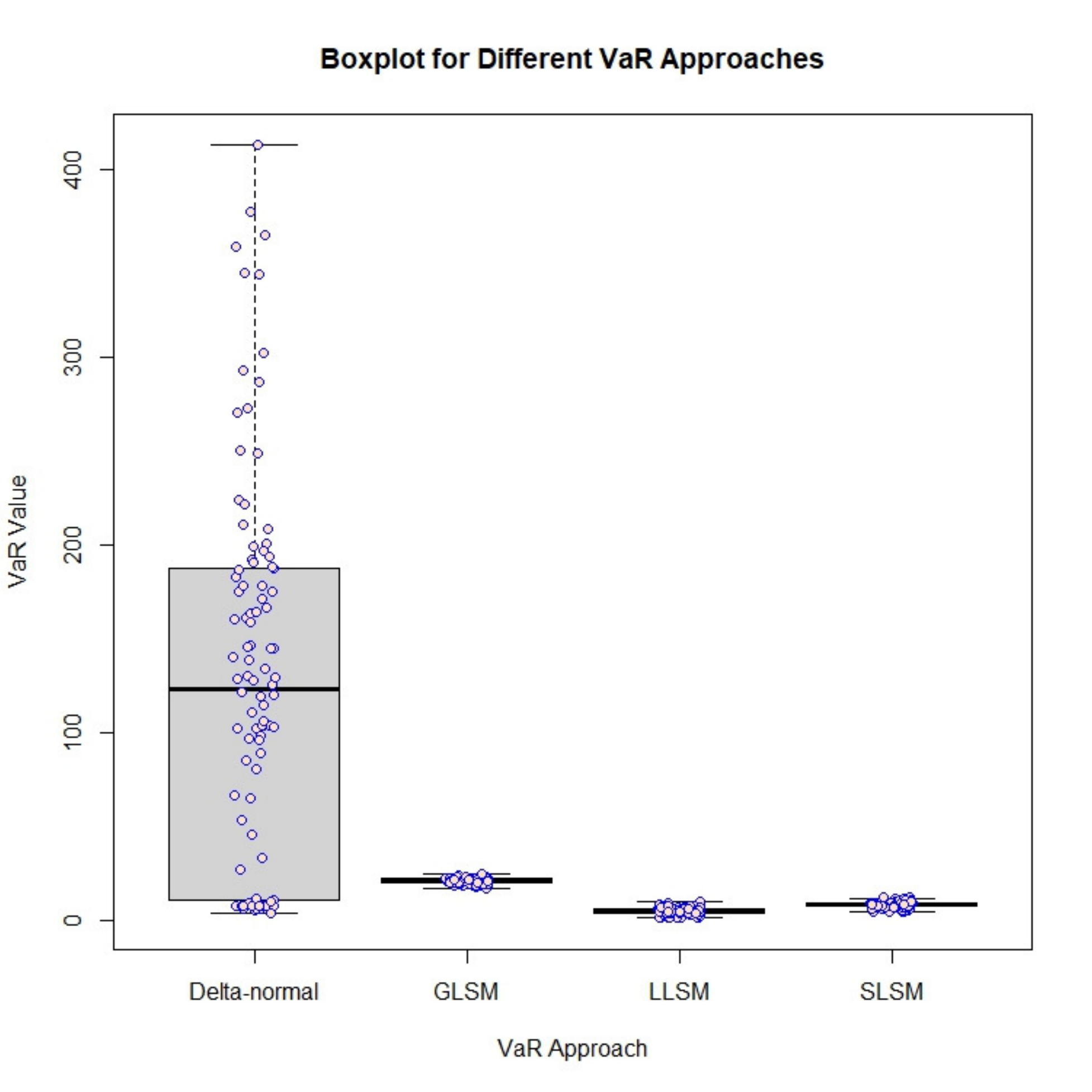}
     \subcaption{VaR boxplot for four approaches}
   \end{minipage}
   \hfill
   \begin{minipage}[b]{.5\linewidth}
     \centering
     \includegraphics[width=1\textwidth]{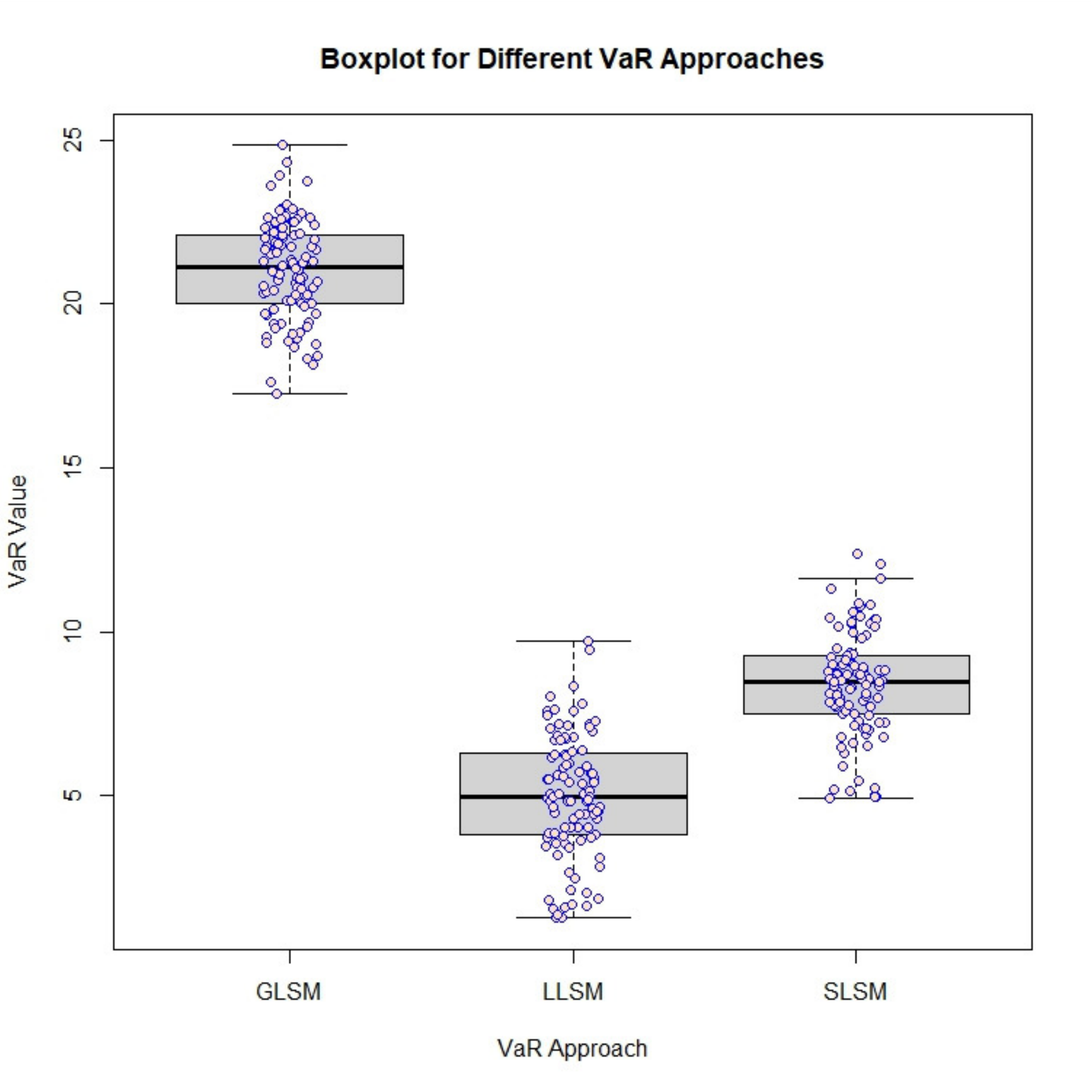}
     \subcaption{VaR boxplots for Bermudan Swaption}
   \end{minipage}
\caption{VaR Boxplot for Bermudan Swaption}
\label{fig:Bboxplot}
\end{figure}

\begin{table}[]\footnotesize
\centering
\caption{Value of Bermudan Swaption}
\label{table:value of Bermudan swaption}
\begin{tabular}{llll|lll}
\hline
    Time & \multicolumn{3}{c|}{$T_2$=year 2} & \multicolumn{3}{c}{$t_1$=day 10}\\

    & Mean & Median & SD
    & Mean & Median & SD  \\
    \hline
    SLSM     & 72.975    &72.946    &   0.85212 & 69.517 & 69.489 & 0.81175 \\
    GLSM     & 75.391    &72.946    &   0.86740 & 71.819 & 71.707 & 0.82638 \\
    LLSM     & 73.452    &73.426    &   0.87146 & 69.971 & 69.946 & 0.83018 \\
    \hline
\end{tabular}
\end{table}

In order to further investigate different performances of the approaches in estimating VaR, we examine valuation performance at the first tenor $T_2$ and $t_1$ and present the result in Table \ref{table:value of Bermudan swaption}. The delta-normal approach is excluded as it does not involve pricing the swaption at $t_1$ and $T_2$. Table \ref{table:value of Bermudan swaption} shows that the valuation at $T_2$ varies little among different approaches. This can be explained by Theorem \ref{thm: valuation convergency}, as well as the analytical result in \cite{Clement_etal-2001-FS}. Consistent with the belief that the fitted value of the regression with OLS estimators at $t_1$ deteriorates, the valuation of GLSM at $t_1$ is significantly different from other three approaches, which is probably an indication of poor valuation estimates at $t_1$. It is also worth mentioning that, as reported in Table \ref{table:value of Bermudan swaption}, the mean values of the swaption prices due to SLSM are close to those evaluated via LLSM. The variables selected by SLSM are chosen by experts with domain knowledge whereas LLSM can automatically include important variables in the regression model amongst a general pool of (polynomials of) covariates in an objective manner. For complicated/new products which are comprised of a vast number of underlying assets, it can be challenging even for practitioners to decide which covariates should be included in the pricing model; the LLSM procedure, on the other hand, can provide hints about which variables that are influential. In addition, although the mean values of the prices due to SLSM and LLSM agree, the corresponding distributions are different, which lead to different tail quantiles, hence the VaR estimates.

The boxplot on the right panel of Figure \ref{fig:Bboxplot} displays the distribution of VaRs estimated via SLSM, GLSM and LLSM. The difference in the distribution of VaRs based on these four approaches indicates that the model selection component in LLSM indeed has a remarkable impact on the VaR values estimated. While the delta-normal method produces highly volatile VaR estimates in Figure \ref{fig:Bboxplot}, we can also see that the estimate produced by GLSM is substantially higher than that given by LLSM.

It is natural to think that the VaR for vanilla equity options should be larger as the number of available stopping times increases. However, the actual relation between VaR and the number of stopping times is more sophisticated for swaptions because their payoff functions that are determined by a large number of dependent underlying forward rate processes. We, therefore, present Table \ref{table:VaR trend} which shows a decreasing VaR trend against the increase in the number of stopping times under our calibrated model. To seek a fair comparison, we adopt the same approach to estimate both the initial value and swaption values at $t_1$ in each column. Based on the decreasing trend observed, one may deduce that Bermudan swaption VaRs should be smaller than those of the oracle VaR of European swaptions. In Table \ref{table:VaR of Bermudan swaption}, only LLSM produces VaR estimates smaller than the oracle VaR of European swaption in Table \ref{table:VaR of European swaption}. Even there is no oracle benchmark for the study of Bermudan swaption, this observation, combined with the possible indication of poor valuation in GLSM and volatile estimates of the delta-normal approach, can justify that for the Bermudan case, LLSM still outperforms other contenders.

\begin{table}[]\footnotesize
\centering
\caption{VaR Trend for Increasing Number of Stopping Times}
\label{table:VaR trend}
\begin{tabular}{llll} \hline
Stopping Times & SLSM        & GLSM        & LLSM        \\ \hline
1              & 8.79705  & 19.7843  & 6.87493  \\
4              & 8.54725  & 21.9728  & 6.06602  \\
6              & 8.31433  & 22.8349  & 5.69590  \\
8              & 8.17318  & 22.9505  & 5.36403  \\
10             & 8.04154  & 22.9719  & 5.08469  \\
12             & 7.92492  & 22.7725  & 4.97697  \\
14             & 7.81349  & 22.6823  & 4.91529  \\
16             & 7.76497  & 22.6077  & 4.84343  \\
18             & 7.75723  & 22.5957  & 4.81491  \\ \hline
\end{tabular}
\end{table}

\section{Conclusion} \label{st:conclusion}
In this paper, we propose the LASSO Least-sqaures Monte Carlo (LLSM) approach as an extension of the Least-squares Monte Carlo (LSM) method for Value-at-Risk (VaR) evaluation of a portfolio. The introduction of LASSO in LLSM, which serves as a model selection technique, enables the proposal to handle high-dimensional and nonlinear portfolios with American features. While domain knowledge facilitates practitioners to select the influential risk factors with more confidence, LLSM offers an objective alternative which can be helpful especially for evaluating VaRs of new and complicated financial products. In this paper, we have also established the oracle properties of LLSM and developed convergence results for pricing and VaR evaluation. Numerical studies in rainbow options and swaptions show that LLSM outperforms other existing practices such as the delta-normal, delta-gamma approaches and LSM.

\textcolor{black}{Although expected shortfall (ES), as a coherent risk measure (see, for instance, \citealp{Gourieroux_Jasiak-2002}), will be implemented in Basel III, we would like to emphasize that an accurate, reliable estimate of VaR is an essential intermediate step for a sound ES estimation. Despite the fact that VaR will play a comparatively lesser role in risk management for the banking industry, it should be stressed that Solvency II, which is the current supervisory framework that has been enforced since 2016 for the insurance industry, makes use of VaR to calculate solvency capital requirement (SCR). On the other hand, as discussed in \cite{Kou_Peng-2016-OR}, the only type of risk measures that satisfy a set of economic axioms for the Choquet expected utility and the statistical property of general elicitability (i.e., there exists an objective function such that minimizing the expected objective function yields the risk measure) is the median shortfall, which is the median of tail loss distribution and is equivalent to the VaR at a higher confidence level. The use of VaR, therefore, does have its merits.}

\textcolor{black}{There are several possible extensions to this paper. Firstly, it is plausible to include historical simulation (HS) or filtered historical simulation (FHS), which are common practices in computing capital requirements in banking industry; see, for example, \cite{Gurrola-Perez_Murphy-2015}, in our framework. Secondly, our discussion on VaR can also be extended to ES. Dantzig selector (see \citealp{Candes_Tao-2007-AoS}) can also shown to be another feasible variable selection method. We shall discuss the corresponding treatment in a separate paper. Thirdly, } 

since the bias term dominates the inaccuracy of LLSM, we can reduce the estimation bias via an extra layer of extensive simulation. As $100(1-\alpha)\%$ $t_1$ VaR is directly affected by the estimate of the $\alpha$ smallest $U_{t_1}$, a more accurate estimate of the quantile will be helpful to improve the performance of LLSM. After getting estimates of $U_{t_1}$ for $N$ scenarios, we can perform intensive simulation to obtain a more accurate estimate of the $\alpha$ smallest $U_{t_1}$. This can be done by first finding the values of underlying assets corresponding to the $\alpha$ smallest estimate of $U_{t_1}$ as initialization, then intensively simulate $N_2$ sample paths under $\mathbb{Q}$ measure. A better estimate of the $\alpha$ smallest $U_{t_1}$ can be found by averaging the discounted payoffs at maturity. We have obtained promising preliminary results for this so-called the Intensive Lasso Least-squares Monte Carlo (ILLSM) approach. Further investigations will be discussed in a separate paper.

\section*{Acknowledgement}
The authors would like to thank the editor, associate editor and the two anonymous referees for their constructive comments that substantially improve the manuscript. The second author is in part financially supported by Hong Kong Research Grant Council research grants ECS-24300514 and GRF-14317716. 

\section*{Appendix A: Proofs of the convergence results}
This appendix contains the proofs for the convergence results discussed in Sections \ref{sst:valuation proof} and \ref{sst:VaR proof}.

\subsection*{A.1 Proof of Theorem \ref{thm: valuation convergency}}\label{sect:A1}

To prove Theorem \ref{thm: valuation convergency}, we need the following four lemmas.
\begin{lemma} \label{lemma 1.2}
Consider a linear regression model $Y=X^\top a+\varepsilon.$ If we have $n$ observations, let
$y=(y_1,\ldots,y_n)^\top$, $y^{m}=(y^{m}_1,\ldots,y^{m}_n)^\top$, $x_i=(x_{1i},\ldots,x_{pi})^\top$, $x=(x_1,\ldots,x_n)$, $x^{(j)}=(x_{j1},x_{j2},\ldots,x_{jn})^\top$, $a=(a_1,\ldots,a_p)^\top$, $\varepsilon=(\varepsilon_1,\ldots,\varepsilon_n)^\top$. $x_i$, $y_i$, $y_i^{m}$ are realizations of random variables $X$, $Y$, $Y^{m}$, where $i=1,\ldots,n$. Define
$$\hat{a}_n^{m}:=\argmin_{\alpha \in \rm I\!R^P}\{ \sum_{i=1}^n(y_i^{[M]}-x_i^\top\alpha)^2+\lambda\|\alpha\|_1\}^2.$$
Assume $\varepsilon_1,\ldots,\varepsilon_n$ are i.i.d. with $E\varepsilon_1=0$, $E|\varepsilon_1|<\infty$, $y_i^{[m]}\overset{a.s.} \to y_i$ as $m \to \infty$. If there exists a non-singular matrix $C$ such that $\frac{1}{n}\sum_{i=1}^n x_ix_i^\top \to C$ as $n \to \infty$,  $\frac{\lambda}{n} \to 0$, then $\hat{a}_n^{m} \overset{a.s.} \to a$ as $n \to \infty$ and $m \to \infty$.
\end{lemma}

\noindent \begin{proof}
Recall that
\begin{align*}
\hat{a}_n^{m}&=\argmin_{\alpha \in \rm I\!R^P}\big \{\sum_{i=1}^n (y_i^{m}-y_i+y_i-x_i^\top a+x_i^\top a-x_i^\top\alpha)^2+\lambda\|\alpha\|_1 \big \} \\
               &=\argmin_{\alpha \in \rm I\!R^P}\big \{\sum_{i=1}^n (y_i^{m}-y_i+\varepsilon_i+x_i^\top(a-\alpha))^2+\lambda\|\alpha\|_1 \big \}.
\end{align*}
Hence, one can write
\begin{eqnarray*}
\hat{a}_n^{m}-a & = & \argmin_{u \in \rm I\!R^P}\left\{\sum_{i=1}^n ((y_i^{m}-y_i)^2+\varepsilon_i^2+(x_i^\top u)^2+2\varepsilon_i(y_i^{m}-y_i)\right.\\
                  &   & \left.-2(y_i^{m}-y_i)x_i^\top u-2\varepsilon_ix_i^\top u)+\lambda\|u+a\|_1\right\}.
\end{eqnarray*}

Define $C_n=\frac{1}{n}\sum_{i=1}^n x_ix_i^\top $, $W_n=\frac{1}{n}\sum_{i=1}^n x_i\varepsilon_i$, $V_n=\frac{1}{n}\sum_{i=1}^n x_i(y_i^{m}-y_i)$ and discard terms which do note involve $u$, we get
\begin{eqnarray*}
\hat{a}_n^{m}-a&=&\argmin_{u \in \rm I\!R^P} \big \{ u^\top C_nu-2W_n^\top u-2V_n^\top u+\frac{\lambda}{n}(\|u+a\|_1-\|a\|_1)   \big \} \\
& \buildrel\triangle\over = & \argmin_{u \in \rm I\!R^P} f_n(u).
\end{eqnarray*}
Let $\gamma_{0,n}$ to be the smallest eigenvalue of $C_n$, $\gamma_0$ to be the smallest eigenvalue of $C$. Then $\gamma_{0,n} \to \gamma_0$ as $n \to \infty$, where $\gamma_0>0$.
Write $\|u\|=\sqrt{\sum_{j=1}^pu_j^2}=\|u\|_2$, which is equivalent to $\ell_2$ norm. If we define
\begin{align*}
\mathscr{T}&=\big \{ \max_{1\leq j\leq p}\frac{1}{n}|x^{(j)T}\varepsilon| \leq \lambda_0   \big \}=\big \{ \max_{1\leq j\leq p} \frac{1}{n} |\sum_{i=1}^n x_{ji}\varepsilon_i| \leq \lambda_0   \big \}, \\
\mathscr{T}_2&=\big \{ \max_{1\leq j\leq p}\frac{2}{n}|x^{(j)T}(y^{m}-y)| \leq \varepsilon^*   \big \}=\big \{ \max_{1\leq j\leq p} \frac{2}{n} |\sum_{i=1}^n x_{ji}(y_i^{m}-y_i)| \leq \varepsilon^*   \big \},
\end{align*}
then on the set $\mathscr{T} \cap \mathscr{T}_2$, we have
\begin{align*}
W_n^\top u=\frac{1}{n}(\sum_{i=1}^nx_i\varepsilon_i)^\top u \leq \lambda_0\sqrt{p} \|u\|,\quad
&V_n^\top u \leq \varepsilon^*\sqrt{p}\|u\|, \\
u^\top C_nu\geq \gamma_{0,n}\|u\|^2, \quad
&\frac{\lambda}{n}(\|u+a\|_1-\|a\|_1)\leq \frac{\lambda}{n}\|u\|_1 \leq \frac{\lambda}{n}\sqrt{p}\|u\|.
\end{align*}
It follows that
\begin{eqnarray*}
f_n(u) &\geq & \gamma_{0,n}\|u\|^2-2\lambda_0\sqrt{p}\|u\|-2\varepsilon^*\sqrt{p}\|u\|-\frac{\lambda}{n}\sqrt{p}\|u\|\\
& = & \|u\|(\gamma_{0,n}\|u\|-2\lambda_0\sqrt{p}-2\varepsilon^*\sqrt{p}-\frac{\lambda}{n}\sqrt{p}).
\end{eqnarray*}
Fix $\lambda_0 \in (0,1)$, $\varepsilon^* \in (0,1)$. Since $\frac{\lambda}{n}=o(1)$ and 
by Lemma 3.1 of \cite{Chatterjee_Lahiri-2011-Sankhy}
, $\frac{1}{n}\sum_{i=1}^nx_i\varepsilon_i \overset{p}\to 0$, there exists $n_0$ such that $\forall n \geq n_0$, $\frac{\lambda}{n} \leq \lambda_0$, $\gamma_{0,n}>\frac{1}{2}\gamma_0>0$.
On the set $\mathscr{T} \cap \mathscr{T}_2$, for any $u \in \rm I\!R^P$ with $\|u\|>\frac{(6\lambda_0+4\varepsilon^*)\sqrt{p}}{\gamma_{0,n}}$, it follows that
\begin{align*}
f_n(u)&\geq \|u\|(\gamma_{0,n}\|u\|-2\lambda_0\sqrt{p}-2\varepsilon^*\sqrt{p}-\lambda_0\sqrt{p})\geq \gamma_{0,n}\frac{\|u\|^2}{2} > 0.
\end{align*}
Since $f_n(0)=0$, it follows that for $n\geq n_0$, the minimum of $f_n(0)$ cannot be attained in the set $\{u:\enspace \|u\|>\frac{(6\lambda_0+4\varepsilon^*)\sqrt{p}}{\gamma_{0,n}}\}$, whenever $\mathscr{T} \cap \mathscr{T}_2$ holds. Hence, $\forall n\geq n_0$, $\mathscr{T} \cap \mathscr{T}_2$ implies that
\begin{align*}
\hat{a}_n^{m}-a&=\argmin_{u}f_n(u) \in \{u:\enspace \|u\|\leq\frac{(6\lambda_0+4\varepsilon^*)\sqrt{p}}{\gamma_{0,n}}\}.
\end{align*}
In particular,
\begin{eqnarray*}
&& \sum_{m=1}^\infty \Pr \left\{\|\hat{a}_n^{m}-a\|>\frac{(6\lambda_0+4\varepsilon^*)\sqrt{p}}{\gamma_{0,n}}~~i.o.\right\}\\
& \leq & \sum_{m=1}^\infty \Pr\{(\mathscr{T} \cap \mathscr{T}_2^{m})^c~~i.o. \} \\
& \leq & \sum_{m=1}^\infty \Pr\{\mathscr{T}^c~~i.o.\}+\sum_{m=1}^\infty \Pr\{(\mathscr{T}_2^{m})^c ~~i.o.\} \\
& = & \sum_{m=1}^\infty \Pr\{(\mathscr{T}_2^{m})^c ~~ i.o.\} <\infty.
\end{eqnarray*}
Since $\lambda_0$ and $\varepsilon^* \in (0,\infty)$ are arbitrary, the proof is completed.
\end{proof}~\\

\begin{lemma} \label{lemma 1.3}
If, for $k=j,\ldots,L-1$
, $a_k^{[M,N]} \overset{a.s.} \to a_k^{[M]}$ as $N \to \infty$ and $\Pr\{a_k^{[M]}\cdot L^{[M]}(X_k)=Z_k\}=0$, then for $i=1,2,\ldots,N$, $Z_{\tau_j^{[i,M,N]}}^{[i]} \overset{a.s} \to Z_{\tau_j^{[i,M]}}^{[i]}$.
\end{lemma}

\begin{proof}
For $j=L$, $Z_{\tau_T^{[i,M,N]}}^{[i]}=Z_{\tau_T^{[i,M]}}^{[i]}=Z_T^{[i]}$. Proceed by induction on j. Assume for $k=j+1, \cdot \cdot \cdot,T-1$, $Z_{\tau_k^{[i,M,N]}}^{[i]} \overset{a.s} \to Z_{\tau_k^{[i,M]}}^{[i]}$, we want to prove $Z_{\tau_j^{[i,M,N]}}^{[i]} \overset{a.s} \to Z_{\tau_j^{[i,M]}}^{[i]}$.
\begin{align*}
& \sum_{N=1}^\infty \Pr \{|Z_{\tau_j^{[i,M,N]}}^{[i]} -Z_{\tau_j^{[i,M]}}^{[i]}| < \varepsilon \}\\
  &\leq \sum_{N=1}^\infty \Pr\{ |Z_{\tau_{j+1}^{[i,M,N]}}^{[i]} -Z_{\tau_{j+1}^{[i,M]}}^{[i]}| < \varepsilon\} \\
  &+\sum_{N=1}^\infty \indicator_{\{a_j^{[M]}\cdot L^{[M]}X_j^{[i]})\leq Z_j^{[i]}<a_j^{[M,N]}\cdot L^{[M]}(X_j^{[i]})\}} \\
  &+\sum_{N=1}^\infty \indicator_{\{a_j^{[M,N]}\cdot L^{[M]}(X_j^{[i]})\leq Z_j^{[i]}<a_j^{[M]}\cdot L^{[M]}(X_j^{[i]})\}} \\
  &<\infty
\end{align*}
because the first term is finite by induction. The second term is bounded by
$$\sum_{N=1}^\infty \indicator_{\{|Z_j^{[i]}-a_j^{[M]}\cdot L^{[M]}(X_j^{[i]})|\leq |(a_j^{[M,N]}-a_j^{[M]})\cdot L^{[M]}(X_j^{[i]})|\}},$$
which is also finite as $\Pr\{Z_j^{[i]}-a_j^{[M]}\cdot L^{[M]}(X_j^{[i]})=0\}=0$. Similarly, the third term can be proved to be finite. This completes the induction. Therefore, as $N \to \infty$, $Z_{\tau_j^{[i,M,N]}}^{[i]} \overset{a.s} \to Z_{\tau_j^{[i,M]}}^{[i]}$
\end{proof}

\begin{lemma}\label{lemma 1.4}
Assume for $j=1,2,\ldots,L-1$, $\Pr\{a_j^{[M]}\cdot L^{[M]}(X_j)=Z_j\}=0$. Furthermore, Conditions (A1)-(A4) are satisfied. Then, for the LASSO estimators $a_j^{[M,N]}$ with penalty parameter $\lambda$ such that $\lambda/N=o(1)$, we have $a_j^{[M,N]} \overset{a.s.} \to a_j^{[M]}$ as $N \to \infty$.
\end{lemma}

\begin{proof}
By Lemma \ref{lemma 1.2}, for $j=L-1$, $a_j^{[M,N]} \overset{a.s.} \to a_j^{[M]}$. We again proceed by induction on j. Assume for $k=j,\cdot \cdot \cdot, T-1$, $a_k^{[M,N]} \overset{a.s.} \to a_k^{[M]}$, our goal is to prove that for $k=j-1$, we still have $a_{j-1}^{[M,N]} \overset{a.s.} \to a_{j-1}^{[M]}$. By Lemma \ref{lemma 1.2}, it suffices to prove for fixed $i=1,2,\cdot \cdot \cdot, N$, as $N \to \infty$, $Z_{\tau_j^{[i,M,N]}}^{[i]} \overset{a.s.} \to Z_{\tau_j^{[i,M]}}^{[i]}.$

By definition, one can write
\begin{eqnarray*}
Z_{\tau_j^{[i,M,N]}}^{[i]}&=&F_j(a_j^{[M,N]},Z^{[i]},X^{[i]}) \\
                          &=&Z_j^{[i]}\indicator_{ \{Z_j^{[i]}\geq a_j^{[M,N]}\cdot L^{[M]}(X_j^{[i]})\}}+Z_{j+1}^{[i]}\indicator_{ \{Z_j^{[i]}< a_j^{[M,N]}\cdot L^{[M]}(X_j^{[i]})\}};\\ &\\
Z_{\tau_j^{[i,M]}}^{[i]}&=&F_j(a_j^{[M]},Z^{[i]},X^{[i]}) \\
                          &=&Z_j^{[i]}\indicator_{ \{Z_j^{[i]}\geq a_j^{[M]}\cdot L^{[M]}(X_j^{[i]})\}}+Z_{j+1}^{[i]}\indicator_{ \{Z_j^{[i]}< a_j^{[M]}\cdot L^{[M]}(X_j^{[i]})\}}\quad\text{and}\\&\\
Z_{\tau_j^{[i,M,N]}}^{[i]}-Z_{\tau_j^{[i,M]}}^{[i]}&=&Z_j^{[i]}\big (\indicator_{\{Z_j^{[i]}\geq a_j^{[M,N]}\cdot L^{[M]}(X_j^{[i]})\}}-\indicator_{\{Z_j^{[i]}\geq a_j^{[M]}\cdot L^{[M]}(X_j^{[i]})\}}\big ) \\
  &&+Z_{\tau_{j+1}^{[i,M,N]}}^{[i]}\indicator_{\{Z_j^{[i]}<a_j^{[M,N]}\cdot L^{[M]}(X_j^{[i]})\}}-Z_{\tau_{j+1}^{[i,M]}}^{[i]}\indicator_{\{Z_j^{[i]}<a_j^{[M]}\cdot L^{[M]}(X_j^{[i]})\}}.
\end{eqnarray*}
By considering the following four cases:
\begin{enumerate}[(i)]
   \item
      If $Z_j^{[i]}\geq a_j^{[M,N]}\cdot L^{[M]}(X_j^{[i]})$ and $Z_j^{[i]}\geq a_j^{[i]}\geq a_j^{[M]}\cdot L^{[M]}(x_j^{[i]})$, $|Z_{\tau_j^{[i,M,N]}}^{[i]}-Z_{\tau_j^{[i,M]}}^{[i]}|=0; $	
   \item
      If $Z_j^{[i]}< a_j^{[M,N]}\cdot L^{[M]}(X_j^{[i]})$ and $Z_j^{[i]}\geq a_j^{[i]}< a_j^{[M]}\cdot L^{[M]}(x_j^{[i]})$, $|Z_{\tau_j^{[i,M,N]}}^{[i]}-Z_{\tau_j^{[i,M]}}^{[i]}|=|Z_{\tau_{j+1}^{[i,M,N]}}^{[i]}-Z_{\tau_{j+1}^{[i,M]}}^{[i]}|;$
   \item
      If $a_j^{[M]}\cdot L^{[M]}(X_j^{[i]}) \leq Z_j^{[i]}<a_j^{[M,N]}\cdot L^{[M]}(X_j^{[i]})$,
$|Z_{\tau_j^{[i,M,N]}}^{[i]}-Z_{\tau_j^{[i,M]}}^{[i]}|=|Z_j^{[i]}-Z_{\tau_{j+1}^{[i,M,N]}}^{[i]}|;$
   \item
      If  $a_j^{[M,N]}\cdot L^{[M]}(X_j^{[i]}) \leq Z_j^{[i]}<a_j^{[M]}\cdot L^{[M]}(X_j^{[i]})$,
$|Z_{\tau_j^{[i,M,N]}}^{[i]}-Z_{\tau_j^{[i,M]}}^{[i]}|=|Z_j^{[i]}-Z_{\tau_{j+1}^{[i,M]}}^{[i]}|,$
\end{enumerate}

we can write
\begin{eqnarray*}
   \sum_{N-1}^\infty \Pr\{|Z_{\tau_j^{[i,M,N]}}^{[i]}-Z_{\tau_j^{[i,M]}}^{[i]}|>\varepsilon\} 
  &\leq &\sum_{N-1}^\infty \Pr\{|Z_{\tau_{j+1}^{[i,M,N]}}^{[i]}-Z_{\tau_{j+1}^{[i,M]}}^{[i]}|>\varepsilon\} \\
  &&+\sum_{N=1}^\infty\indicator_{\{a_j^{[M]}\cdot L^{[M]}(X_j^{[i]})\leq Z_j^{[i]}<a_j^{[M,N]}\cdot L^{[M]}(X_j^{[i]})\}} \\
  &&+\sum_{N=1}^\infty\indicator_{\{a_j^{[M,N]}\cdot L^{[M]}(X_j^{[i]})\leq Z_j^{[i]}<a_j^{[M]}\cdot L^{[M]}(X_j^{[i]})\}} \\
  &\buildrel\triangle\over = &I_1+I_2+I_3.
\end{eqnarray*}
By Lemma \ref{lemma 1.3} and $a_{j+1}^{[M,N]} \overset{a.s.} \to a_{j+1}^{[M]}$, $I_1<\infty$.
\begin{align*}
I_2+I_3 &\leq \sum_{N=1}^\infty \indicator_{\{|Z_j^{[i]}-a_j^{[M]}\cdot L^{[M]}(X_{j+1}^{[i]})|\leq |a_j^{[M,N]}-a_j^{[M]}||L^{[M]}(X_j^{[i]})|\}} <\infty.
\end{align*}
Since $a_j^{[M,N]} \overset{a.s.} \to a_j^{[M]}$, $\Pr\{Z_j=a_j^{[M]}\cdot L^{[M]}(X_j)\}=0$, we conclude that  $Z_{\tau_j^{[i,M,N]}}^{[i]} \overset{a.s.} \to Z_{\tau_j^{[i,M]}}^{[i]}$. This completes the induction.
\end{proof}

\begin{lemma} \label{lemma 1.5}
Consider a linear regression model: $Y=X^\top a+\epsilon$. If we have $n$ observations, let
$y=(y_1,\ldots,y_n)^\top$, $x_i=(x_{1i},\ldots,x_{pi})^\top$, $x=(x_1,\ldots,x_n)$, $x^{(j)}=(x_{j1},x_{j2},\ldots,x_{jn})^\top$, $a=(a_1,\ldots,a_p)^\top$, $\varepsilon=(\varepsilon_1,\ldots,\varepsilon_n)^\top$. We also define
$$\hat{a}_n^m:=\argmin_{\alpha \in \rm I\!R^P}\left( \sum_{i=1}^n(y_i^m-x_i^\top\alpha)^2+\lambda\|\alpha\|_1\right)$$
and denote the true parameters in the regression model by $a$.
\noindent Assume $\varepsilon_1,\ldots,\varepsilon_n$ are i.i.d. with $E\varepsilon_1=0$, $E|\varepsilon_1|<\infty$, $y_i^m\overset{a.s.} \to y_i$ as $m \to \infty$. If the compatibility condition holds for $S_0$ and $\lambda$ is a suitable penalty parameters satisfying $\lambda/n \to 0$ and $\lambda=\mathcal{O}(\log p/n)$, then $\hat{a}_n^m \overset{a.s.} \to a $ as $n \to \infty$ and $m \to \infty$.
\end{lemma}

\begin{proof}
The proof is similar to that of Lemma \ref{lemma 1.2}. We adopt same notation used in Lemma \ref{lemma 1.2} and omit some part of the proof.  Again, observing that
$$
W_n^\top u\leq \lambda_0\sqrt{p}\|u\|,\quad
V_n^\top u\leq \varepsilon^*\sqrt{p}\|u\| \quad\text{and}\quad
u^\top C_n u\geq \|u_{S_0}\|_1^2\frac{\phi_0^2}{s_0}>0,
$$
we can write
\begin{align*}
f_n(u)&\geq \|u_{S_0}\|_1^2\frac{\phi_0^2}{s_0}-2\lambda_0\sqrt{p}\|u\|-2\varepsilon^*\sqrt{p}\|u\|-\frac{\lambda}{n}\sqrt{p}\|u\| \\
      &\geq \|u_{S_0}\|(\frac{\phi_0^2}{s_0}\|u_{S_0}\|-2\lambda_0\sqrt{p}-2\varepsilon^*\sqrt{p}-\frac{\lambda}{n}\sqrt{p}).
\end{align*}
Fix $\lambda_0 \in (0,1)$, $\varepsilon^* \in (0,1)$. Since $\lambda/n=o(1)$, there exists $n_0$ such that $\forall n\geq n_0$, $\lambda/n\leq \lambda_0$.

On the set $\mathscr{T} \cap \mathscr{T}_2$, $\forall u \in \rm I\!R^P$ with $\|u_{S_0}\|>\frac{(6\lambda_0+4\varepsilon^*)\sqrt{p}}{\phi_0^2/s_0}$,
\begin{align*}
f_n(u)&\geq\|u_{S_0}\|(\frac{\phi_0^2}{s_0}\|u_{S_0}\|-2\lambda_0\sqrt{p}-2\varepsilon^*\sqrt{p}-\lambda_0\sqrt{p}) \geq \frac{\phi_0^2}{s_0}\frac{\|u_{S_0}\|^2}{2}>0.
\end{align*}
Since $f_n(0)=0$, it follows that for $n\geq n_0$, the minimum of $f_n(0)$ cannot be obtained in the set $\{u:\|u_{S_0}\|>\frac{(6\lambda_0+4\varepsilon^*)\sqrt{p}}{\phi_0^2/s_0}\}$, whenever $\mathscr{T} \cap \mathscr{T}_2$ holds. Hence, for $n\geq n_0$, $\mathscr{T} \cap \mathscr{T}_2$ implies
\begin{align*}
\hat{a}_n^{[M]}-a&=\argmin_u f_n(u)\in \{u: \|u_{S_0}\|\leq \frac{(6\lambda_0+4\varepsilon^*)\sqrt{p}}{\phi_0^2/s_0}\}.
\end{align*}
Due to the Compatibility Condition, we can write
\begin{align*}
\|u\| &\leq \|u_{S_0}\|+\|u_{S_0^c}\| \leq 10\|u_{S_0}\|
\end{align*}
because $\|u_{S_0^c}\|_1\leq 3\|u_{S_0}\|_1$ implies $\|u_{S_0^c}\|\leq 9\|u_{S_0}\|$.
As a result,
\begin{eqnarray*}
   &&\sum_{M=1}^{\infty} \Pr\{\|\hat{a}_n^{[M]}-a\|>\frac{10(6\lambda_0+4\varepsilon^*)\sqrt{p}}{\phi_0^2/s_0}~~ i.o.\}\\
  & \leq & \sum_{M=1}^{\infty} \Pr\{\|u\|>\frac{10(6\lambda_0+4\varepsilon^*)\sqrt{p}}{\phi_0^2/s_0}~~ i.o.\} \\
  & \leq & \sum_{M=1}^{\infty} \Pr\{\|u_{S_0}\|+\|u_{S_0^c}\|>\frac{10(6\lambda_0+4\varepsilon^*)\sqrt{p}}{\phi_0^2/s_0}~~ i.o.\} \\
  & \leq & \sum_{M=1}^{\infty} \Pr\{10\|u_{S_0}\|>\frac{10(6\lambda_0+4\varepsilon^*)\sqrt{p}}{\phi_0^2/s_0}~~ i.o.\} \\
  & \leq & \sum_{M=1}^{\infty} \Pr\{(\mathscr{T}\cap\mathscr{T}_2^{[M]})^c ~~ i.o.\} < \infty.
\end{eqnarray*}
Since $\lambda_0$ and $\varepsilon^* \in (0,\infty)$ are arbitrary, this completes the proof.
\end{proof}

\begin{proof}[Proof of Theorem \ref{thm: valuation convergency}]
The proof of Theorem \ref{thm: valuation convergency} (i) can be established based on preceding lemmas \ref{lemma 1.2}-\ref{lemma 1.5}. It is equivalent to prove
$$\lim_{N \to \infty} \frac{1}{N} \sum_{i=1}^N U_j^{[i,M,N]}=\mathbf{E}(U_j^{M}|\mathscr{F}_j).$$
By the Law of large numbers (LLNs), it suffices to prove 
$$G_N\buildrel\triangle\over = \frac{1}{N}\sum_{i=1}^N \left(U_j^{[i,M,N]}-U_j^{[i,M]}\right)$$

By Lemma 3.1 of \cite{Clement_etal-2001-FS}, we can write
\begin{align*}
|G_N| &\leq \frac{1}{N}\sum_{i=1}^N  \left|U_j^{[i,M,N]}-U_j^{[i,M]}\right| \\
      &\leq \frac{1}{N}\sum_{i=1}^N \sum_{k=j}^T |z_k^{[i]}| \sum_{k=j}^{T-1} \indicator_{ \{|Z_k^{[i]}-a_k^{[M]}\cdot L^{[M]}(X_k^{[i]})|\leq |(a_k^{[M,N]}-a_k^{[M]})\cdot L^{[M]}(X_k^{[i]})|\}}.
\end{align*}
Since for $j=1,\ldots,L-1$, $a_j^{[M,N]} \overset{a.s.} \to a_j^{[M]}$. Then $\forall \varepsilon>0$,
\begin{align*}
\limsup_N|G_N| &\leq \limsup_N \frac{1}{N} \sum_{i=1}^N \sum_{k=j}^T |Z_k^{[i]}|\sum_{k=j}^{T-1}\indicator_{ \{|Z_k^{[i]}-a_k^{[M]}\cdot L^{[M]}(X_k^{[i]})|\leq |\varepsilon \cdot L^{[M]}(X_k^{[i]})|\}} \\
               &=\mathbf{E} \left\{ \sum_{k=j}^T |Z_k|\sum_{k=j}^{T-1}\indicator_{ \{|Z_k-a_k^{[M]}\cdot L^{[M]}(X_k)|\leq |\varepsilon \cdot L^{[M]}(X_k)| \}}  \right\}.
\end{align*}
The last equality follows from LLN. Let $\varepsilon \to 0$, we obtain the convergence to zero since for $j=1,\ldots,L-1$, $\Pr\{a_j^{[M]}\cdot L^{[M]}(X_j)=Z_j\}=0$. The proof of Theorem \ref{thm: valuation convergency} (ii) follows if we substitute Lemma \ref{lemma 1.5} for Lemma \ref{lemma 1.2} in the preceding proof.
\end{proof}

\subsection*{A.2 Proof of Theorem \ref{thm:valuation advantage}}\label{sect:A2}
To define the irrepresentable condition and relevant active set, we first re-write the gram matrix $A_j^{[M,N]}$ as $A_j$, $c_{k,l}$ is the element in the $k$-th row and $l$-th column in the matrix $A_j$. Define submatrices of the gram matrix $A_j$ given an index set $S$ as
\begin{align*}
&A_{1,1}^{(j)}(S)=(c_{k,l})_{k,l \in S} &&A_{2,2}^{(j)}(S)=(c_{k,l})_{k,l \notin S} \\
&A_{1,2}^{(j)}(S)=(c_{k,l})_{k\in S,l \notin S} &&A_{2,1}^{(j)}(S)=A_{1,2}^{(j)\top}(S).
\end{align*}

The Irrepresentable Condition and the relevant active set are defined as follows: We say that the Irrepresentable Condition is met for the set $S$ with cardinality $s$, if for all vector $u_S \in \rm I\!R^s$ satisfying $\|u_S\|_\infty\leq 1$, we have
$$\|A_{2,1}(S)A_{1,1}^{-1}(S)u_S\|_\infty<1.$$ \\
In addition, relevant active set $S_0^{\text{relevant}}$ is defined as for fixed $j \in \{0,...,T-1\}$, $$S_0^{\text{relevant}}\buildrel\triangle\over=\left\{m:|a_{j,m}^{[M]}|>\lambda^{(j)}\sup_{\|u_{S_0}\|_\infty\leq 1} \|A_{1,1}^{(j)-1}(S_0)u_{S_0}\|_\infty/2 \right\},$$
where $S_0$ is the active set, $a_{j,m}^{[M]}$ is the $m$-th element of the true coefficient vector $a_j^{[M]}$.\\

The following lemma is due to Theorem 7.1 of \cite{Buhlmann_vandeGeer-2011}.
\begin{lemma} \label{lemma:active set in lasso}
Suppose the Irrepresentable Condition holds for $S_0$. Then $S_0^{\text{relevant}}\subset S(\lambda)\subset S_0$ and for $j=0,...,L-1$,
$$\|(a_j^{[M,N]})_{S_0}-(a_j^{[M]})_{S_0}\|_\infty\leq \lambda \sup_{\|u_{S_0}\|\infty\leq 1}\|\Sigma_{1,1}^{(j)-1}(S_0)u_{S_0}\|_\infty /2,$$
where $a_j^{[M,N]}$ is the LASSO estimated coefficients with penalty $\lambda$, $S_0(\lambda)=\{k,a_{j,k}^{[M,N]}\neq 0\}$.
\end{lemma}
\begin{proof}[Proof of Theorem \ref{thm:valuation advantage}]
Our proof skips some steps that are similar to the proof of Theorem 3.1 in \cite{Clement_etal-2001-FS}. It is equivalent to prove for $j=0,\ldots,L$, $$\lim_{N\to \infty}\mathbf{E}(Z_{\tau_{j}^{[M,N]}}|\mathscr{F}_j)=\mathbf{E}(Z_{\tau_j}|\mathscr{F}_j).$$

Note that the following induction holds for both $M_1$ and $M$ until specification. For $j=L$, $\tau_T^{[M,N]}=\tau_T=T$ and $\mathbf{E}(Z_{\tau_{j}^{[M,N]}}|\mathscr{F}_j)=\mathbf{E}(Z_{\tau_j}|\mathscr{F}_j)$. Assume $\lim_{N\to \infty}\mathbf{E}(Z_{\tau_{k}^{[M,N]}}|\mathscr{F}_k) =\mathbf{E}(Z_{\tau_k}|\mathscr{F}_k)$ holds for $k=j+1$, we want to prove it also holds for $k=j$.
\begin{eqnarray*}
\mathbf{E}(Z_{\tau_j^{[M,N]}}|\mathscr{F}_j)
&=&\frac{1}{N}\sum_{i=1}^N Z_{\tau_j^{[i,M,N]}}^{[i]} \\
&=&\frac{1}{N}\sum_{i=1}^N \left[ Z_j^{[i]} \indicator_{\{Z_j^{[i]}\geq a_j^{[M,N]}\cdot L^{[M]}(X_j^{[i]})  \}}+Z_{\tau_{j+1}^{[i,M,N]}}^{[i]}\indicator_{\{Z_j^{[i]}<a_j^{[M,N]}\cdot L^{[M]}(X_j^{[i]}) \}}\right]
\end{eqnarray*}
and
\begin{eqnarray*}
\mathbf{E}(Z_{\tau_j}^{[M,N]}-Z_{\tau_j}|\mathscr{F}_j)
&=&\{Z_j-\mathbf{E}(Z_{\tau_{j+1}}|\mathscr{F}_j)\}(\indicator_{\{Z_j\geq a_j^{[M,N]}\cdot L^{[M]}(X_j)\}}-\indicator_{\{Z_j>\mathbf{E}(Z_{\tau_{j+1}}|\mathscr{F}_j)  \}})  \\
&&+\mathbf{E}(Z_{\tau_{j+1}^{[M,N]}}-Z_{\tau_{j+1}}|\mathscr{F}_j)\indicator_{\{Z_j<a_j^{[M,N]}\cdot L^{[M]}(X_j) \}}.
\end{eqnarray*}

The second term in the RHS converges to zero by induction. Next, observe that
\begin{eqnarray*}
 && |B_j^{[M]}|\\
 & \buildrel\triangle\over = & |(Z_j-\mathbf{E}(Z_{\tau_{j+1}}|\mathscr{F}_j))(\indicator_{\{Z_j\geq a_j^{[M,N]}\cdot L^{[M]}(X_j)\}}-\indicator_{\{Z_j>\mathbf{E}(Z_{\tau_{j+1}}|\mathscr{F}_j) \}})| \\
 & \leq & |Z_j-\mathbf{E}(Z_{\tau_{j+1}}|\mathscr{F}_j)|\indicator_{\{|Z_j-\mathbf{E}(Z_{\tau_{j+1}}|\mathscr{F}_j)|\leq |a_j^{[M,N]}\cdot L^{[M]}(X_j)-\mathbf{E}(Z_{\tau_{j+1}}|\mathscr{F}_j)| \}} \\
 & \leq & |a_j^{[M,N]}\cdot L^{[M]}(X_j)-\mathbf{E}(Z_{\tau_{j+1}}|\mathscr{F}_j)| \\
 & \leq & |a_j^{[M,N]}\cdot L^{[M]}(X_j)-P_j^{[M]}(\mathbf{E}(Z_{\tau_{j+1}}|F_j))|+|P_j^{[M]}(\mathbf{E}(Z_{\tau_{j+1}}|F_j))-\mathbf{E}(Z_{\tau_{j+1}}|F_j)|.
\end{eqnarray*}
By definition of the projection $P_j(\cdot)$,
$$P_j^{[M]}(\mathbf{E}(Z_{\tau_{j+1}^{[M]}}|F_j))=a_j^{[M]}\cdot L^{[M]}(X_j).$$
Therefore, one can write
\begin{align*}
|B_j^{[M_1]}| & \leq |a_j^{[M_1,N]}\cdot L^{[M_1]}(X_j)-a_j^{[M_1]}\cdot L^{[M_1]}(X_j)| \\
 &+ |P_j^{[M_1]}(\mathbf{E}(Z_{\tau_{j+1}^{[M_1]}}|\mathscr{F}_j))-P_j^{[M_1]}(\mathbf{E}(Z_{\tau_{j+1}}|\mathscr{F}_j))| \\
 &+|P_j^{[M_1]}(\mathbf{E}(Z_{\tau_{j+1}}|\mathscr{F}_j))-\mathbf{E}(Z_{\tau_{j+1}}|\mathscr{F}_j)|.
\end{align*}

As $N \to \infty$, the first term in the R.H.S. converges to zero by Theorem \ref{thm:valuation convergence layer1}. The second term is zero by Theorem \ref{thm:valuation convergence layer1} since these $M_1$ basis functions span $L^2\{\sigma(X_j)\}$.
\begin{align*}
|B_j^{[M]}| & \leq |a_j^{[M,N]}\cdot L^{[M]}(X_j)-a_j^{[M]}\cdot L^{[M]}(X_j)|\\
            & + |P_j^{[M]}(\mathbf{E}(Z_{\tau_{j+1}^{[M]}}|\mathscr{F}_j))-P_j^{[M]}(\mathbf{E}(Z_{\tau_{j+1}}|\mathscr{F}_j))| \\
            &+|P_j^{[M]}(\mathbf{E}(Z_{\tau_{j+1}}|\mathscr{F}_j))-\mathbf{E}(Z_{\tau_{j+1}}|\mathscr{F}_j)|.
\end{align*}

As $N \to \infty$, the first term in the R.H.S. converges to zero since Theorem \ref{thm:valuation convergence layer1} is applicable to any fixed $M$. The second term is zero by Theorem \ref{thm:valuation convergence layer1} since these $M_1$ basis functions span $L^2(\sigma(X_j))$. To prove the convergence for the second term, it suffices to prove
\begin{align*}
\bigg|\mathbf{E}(Z_{\tau_{j+1}^{[M]}}|\mathscr{F}_j)-\mathbf{E}(Z_{\tau_{j+1}}|\mathscr{F}_j)\bigg|&=\bigg|\mathbf{E}(Z_{\tau_{j+1}^{[M]}}|\mathscr{F}_j)-\mathbf{E}(Z_{\tau_{j+1}^{[M_1]}}|\mathscr{F}_j)\bigg|\\
    & = \big|(a_j)_{S_0\setminus S_0(\lambda)} \cdot (L(X_j))_{S_0\setminus S_0(\lambda)}\big| \to 0 .
\end{align*}

\begin{enumerate}[(i)]
  \item
     To prove $U_j^{*[M_1,N]} \overset{a.s} \to U_j$, it remains to prove as $N\to \infty$,
$$\big|P_j^{[M_1]}(\mathbf{E}(Z_{\tau_{j+1}}|\mathscr{F}_j))-\mathbf{E}(Z_{\tau_{j+1}}|\mathscr{F}_j)\big|\to 0.$$
  \item
     To prove $U_j^{[M,N]}  \overset{a.s} \to U_j$, it remains to prove as $N\to \infty$,
$$|P_j^{[M]}(\mathbf{E}(Z_{\tau_{j+1}}|\mathscr{F}_j))-\mathbf{E}(Z_{\tau_{j+1}}|\mathscr{F}_j)|\to 0,\quad |(a_j)_{S_0\setminus S_0(\lambda)} \cdot (L_(X_j))_{S_0\setminus S_0(\lambda)}| \to 0 .$$
\end{enumerate}

By Condition (A1),
\begin{align*}
\mathbf{E}(Z_{\tau_{j+1}}|F_j)&=a_{j,1}\cdot L_1(X_j)+\ldots+a_{j,k}\cdot L_k(X_j) = (a_j)_{S_0}\cdot \big (L(X_j) \big )_{S_0}.
\end{align*}

For (i), $P_j^{[M_1]}(\mathbf{E}(Z_{\tau_{j+1}}|F_j))=(a_j)_{S_0^{[M_1]}}\cdot \big (L(X_j) \big )_{S_0^{[M_1]}}$. Recall that $S_0\subseteq S_0^{[M_1]}$. For $k \in S_0 \subseteq S_0^{[M_1]}$, $a_{j,k}^{[M_1]}=a_{j,k}\neq 0$. For $k \in S_0^\mathsf{c} \setminus (S_0^{[M_1]})^\mathsf{c}$,  $a_{j,k}^{[M_1]}=a_{j,k}\neq 0$. It follows that $
	(a_j)_{S_0}\cdot \big (L(X_j) \big )_{S_0} = (a_j)_{S_0^{[M_1]}}\cdot \big (L(X_j) \big )_{S_0^{[M_1]}}$ and $\big|P_j^{[M_1]}(\mathbf{E}(Z_{\tau_{j+1}}|F_j))-\mathbf{E}(Z_{\tau_{j+1}}|F_j)\big|0$.

For (ii), $P_j^{[M]}(\mathbf{E}(Z_{\tau_{j+1}}|F_j))=(a_j)_{S_0(\lambda)}\cdot \big (L(X_j) \big )_{S_0(\lambda)}$. There are $M$ basis functions selected from the initial regression with $M_1$ basis functions by LASSO with penalty $\lambda$ where $M\leq M_1$.
Define $$S_0^{\text{relevant}}\buildrel\triangle\over=\{k:|a_{j,k}^{[M_1]}|>\lambda^{(j)}\sup_{\|u_{S_0}\|_\infty\leq 1} \|\Sigma_{1,1}^{(j)-1}(S_0)u_{S_0}\|_\infty/2 \}$$

Then by Lemma \ref{lemma:active set in lasso}, $S_0^{\text{relevant}}\subseteq S_0(\lambda)\subseteq S_0\subseteq S_0^{[M_1]}$. For $k \in S_0(\lambda) \subseteq S_0$, $a_{j,k}^{[M]}=a_{j,k}\neq 0$. For $k \in S_0\setminus (S_0(\lambda))$,  $a_{j,k}^{[M]}=0$, $a_{j,k}\neq 0$,where $S_0\setminus (S_0(\lambda)) \subseteq S_0 \setminus S_0^{\text{relevant}}=\{k: 0<|a_{j,k}^{[M_1]}|<\lambda^{(j)}\sup_{\|u_{S_0}\|_\infty\leq 1} \|\Sigma_{1,1}^{(j)-1}(S_0)u_{S_0}\|_\infty/2 \}$.

It follows that
\begin{eqnarray*}
  |P_j^{[M]}(\mathbf{E}(Z_{\tau_{j+1}}|\mathscr{F}_j))-\mathbf{E}(Z_{\tau_{j+1}}|\mathscr{F}_j)|
  &=& (a_j)_{S_0\setminus S_0(\lambda)}\cdot \big (L(X_j) \big )_{S_0\setminus S_0(\lambda)} \\
  & \leq &\lambda^{(j)}\left\{\sup_{\|u_{S_0}\|_\infty\leq 1} \|\Sigma_{1,1}^{(j)-1}(S_0)u_{S_0}\|_\infty/2\right\}\bigg|\sum_{k\in S_0\setminus S_0(\lambda)}L_k(X_j)\bigg|  \\
  & \to& 0 \quad \text{ as } N \to \infty
\end{eqnarray*}
Since $\lambda^{(j)}\sup_{\|u_{S_0}\|_\infty\leq 1} \|\Sigma_{1,1}^{(j)-1}(S_0)u_{S_0}\|_\infty/2 \to 0$ as $N \to \infty$. The remaining term $|\sum_{k\in S_0\setminus S_0(\lambda)}L_k(X_j)|<\sum_{k\in S_0\setminus S_0(\lambda)}|L_k(X_j)|<\infty$
since $|S_0|=s_0<\infty$, $|X_j|<\infty$, $|L_k(X_j)|<\infty$ for all $k\in S_0$.
\end{proof}

\subsection*{A.3 Proof of Theorem \ref{thm: VaR convergency}}\label{sect:A3}

\begin{proof}[Proof of Theorem \ref{thm: VaR convergency}]
  We begin the proof by rewriting $\text{VaR}_j^{[M,N]}$, $\text{VaR}_j^{[M]}$ as
$$\Pr\{U_j^{[M,N]}>\text{VaR}_j^{[M,N]}\}=\Pr\{U_j^{[M]}>\text{VaR}_j^{[M]}\}=\alpha^\prime.$$
where $\alpha^\prime$ is a deterministic known constant. By Theorem \ref{thm: valuation convergency}, $U_{t_1}^{[M,N]}\overset{a.s.} \to U_{t_1}^{[M]}$ as $N \to \infty$. Denote the pdf of $U_{t_1}^{[M,N]}$ and $U_{t_1}^{[M]}$ as $g_N(u)$ and $g(u)$ respectively, then
$$\int_{-\infty}^{\text{VaR}_j^{[M,N]}}g_N(u)du=\int_{-\infty}^{\text{VaR}_j^{[M]}}g(u)du=\alpha^\prime.$$
\begin{align*}
0&=\int_{-\infty}^{\text{VaR}_j^{[M]}}g_N(u)du-\int_{-\infty}^{\text{VaR}_j^{[M]}}g(u)du+\int_{\text{VaR}_j^{[M]}}^{\text{VaR}_j^{[M,N]}}g_N(u)du \\
  &=G_n(\text{VaR}_j^{[M]})-G(\text{VaR}_j^{[M]})+\int_{\text{VaR}_j^{[M]}}^{\text{VaR}_j^{[M,N]}}g_N(u)du,
\end{align*}
where $G_N(u)$, $G(u)$ is the cdf of $U_{t_1}^{[M,N]}$, $U_{t_1}^{[M]}$. As $U_{t_1}^{[M,N]}\overset{a.s.} \to U_{t_1}^{[M]}$, we have $U_{t_1}^{[M,N]}\overset{d} \to U_{t_1}^{[M]}$, $G_N(\text{VaR}_j^{[M]}) \to G(\text{VaR}_j^{[M]})$,  $|\int_{\text{VaR}_j^{[M]}}^{\text{VaR}_j^{[M,N]}}g_N(u)du| \to 0$. We complete the proof by contradiction.

Assume $\text{VaR}_j^{[M,N]} \nrightarrow \text{VaR}_j^{[M]}$, then $\forall N \in N_+$, $\exists \epsilon_0>0$, st $|\text{VaR}_j^{[M,N]}-\text{VaR}_j^{[M]}|>\epsilon_0$. As the support set of the distribution of $\text{VaR}_j^{[M,N]}$ is tight, there exists $u_0 \in [\min(\text{VaR}_j^{[M,N]},\text{VaR}_j^{[M]}), \max(\text{VaR}_j^{[M,N]},\text{VaR}_j^{[M]})]$ such that  $g_N(u_0)>0$.

If $U_{t_1}^{[M,N]}$ is discrete,
\begin{align*}
\bigg|\int_{\text{VaR}_j^{[M]}}^{\text{VaR}_j^{[M,N]}}g_n(u)du\bigg|=\bigg|G_N(\text{VaR}_j^{[M]})-G(\text{VaR}_j^{[M]})\bigg|>0,
\end{align*}
contradiction.

If $U_{t_1}^{[M,N]}$ is continuous, $\exists \epsilon_0^*>0$, $\forall u \in (u_0-\epsilon_0^*,u_0+\epsilon_0^*)~\cap~\max(\text{VaR}_j^{[M,N]},\text{VaR}_j^{[M]})]$, $g_N(u)>u_0^*>0$,
$$\bigg|\int_{\text{VaR}_j^{[M]}}^{\text{VaR}_j^{[M,N]}}g_N(u)du\bigg|>u_0^* \min(2\epsilon_0^*,\epsilon_0)>0,$$
contradiction. Therefore, the assumption $\text{VaR}_j^{[M,N]} \nrightarrow \text{VaR}_j^{[M]}$ is not true in which case $\text{VaR}_j^{[M,N]} \to \text{VaR}_j^{[M]}$ as $N \to \infty$.
\end{proof}

\section*{A.4 Proof of Theorem \ref{thm:VaR convergency rate}} \label{sect:A4}
To prove this theorem, we first introduce the following lemma and its proof.

\begin{lemma}
\label{lemma:VaR convergence rate}
Let $\alpha_N \buildrel\triangle\over = \Pr\{U_0-U_{t_1}^{[M,N]}<-\text{VaR}_{t_1}^{[M]} \}$,
$\alpha_N^* \buildrel\triangle\over = \Pr\{U_0-U_{t_1}^{*[M,N]}<-\text{VaR}_{t_1}^{[M]} \}$. Assume conditions in Theorem \ref{thm: valuation convergency}(\text{ii}) are satisfied and Condition (A6) holds for $W$ and $W^*$ respectively, then
\begin{align*}
\alpha_N-\alpha=&-g(v)\mathcal{O}\left(\sqrt{\frac{s_0\log M}{N\phi_0^2}}\right)+\frac{d}{dv}g(v)\mathcal{O}\left(\frac{s_0\log M}{N\phi_0^2}\right)+o\left(N^{-1}\right), \\
\alpha_N^*-\alpha=&\frac{d}{dv}g(v)\mathcal{O}\left(\frac{M}{N}\right)+g(v)o\left(N^{-1}\right),
\end{align*}
where 
$\phi_0$ denotes the compatibility constant defined in the Compatibility Condition.~\\
\end{lemma}

\begin{proof}[Proof of Lemma \ref{lemma:VaR convergence rate}]
Using Taylor expansion, we can write
\begin{eqnarray*}
\alpha_N-\alpha & = & \int_{\rm I\!R}\int_{v+w/\sqrt{N}}^v g_N(u,w)dudw \\
                & = & -\int_{\rm I\!R}\frac{w}{\sqrt{N}}g_N(v,w)dw+\int_{\rm I\!R}\frac{w^2}{2N}\frac{\partial}{\partial v}g_N(v,w)dw+\mathcal{O}\left(\frac{1}{N^{3/2}}\right).
\end{eqnarray*}

The first term can be written as,
\begin{eqnarray*}
  \int_{\rm I\!R}\frac{w}{\sqrt{N}}g_N(v,w)dw
  & = &\frac{g(v)}{\sqrt{N}}\mathbf{E}(W| U_{t_1}^{[M]}=v) \\
  &=& g(v)\mathbf{E}\left\{\mathbf{E}\left[N^{-1}\sum_{i=1}^N(a_{t_1}^{[M,N]}-\tilde{a}_{t_1}^{[M,N]})\cdot L^{[M]}(X_{t_1})| U_{t_1}^{[M]}=v,X_{t_1}\right]\right\} \\
  && + g(v)\mathbf{E}\left\{\mathbf{E}\left[N^{-1}\sum_{i=1}^N(\tilde{a}_{t_1}^{[M,N]}-a_{t_1}^{[M]})\cdot L^{[M]}(X_{t_1})| U_{t_1}^{[M]}=v,X_{t_1}\right]\right\} \\
  &=& g(v)\mathcal{O}\left(\sqrt{\frac{s_0\log M}{N}}\right)+o\left(N^{-1}\right).
\end{eqnarray*}

The last equality follows from Theorem 7.7 in \cite{Buhlmann_vandeGeer-2011} and Theorem \ref{thm: valuation convergency}. Regarding the second term, we can write,

\begin{eqnarray*}
\frac{1}{2N}\int_{\rm I\!R}w^2\frac{\partial}{\partial u}g_N(v,w)dw & = & \frac{1}{2N}\frac{d}{dv}g(v)\mathbf{E}\left\{\mathbf{E}(W^2| X_{t_1}) | U_{t_1}^{[M]}=v\right\} \\
   & = & \frac{d}{dv}g(v)\mathcal{O}\left(\frac{s_0\log M}{N\phi_0^2}\right).
\end{eqnarray*}

It follows that
\begin{equation*}\alpha_N-\alpha=g(v)\mathcal{O}\left(\sqrt{\frac{s_0\log M}{N\phi_0^2}}\right)+\frac{d}{dv}g(v)\mathcal{O}\left(\frac{s_0\log M}{N\phi_0^2}\right)+o\left(N^{-1}\right).\end{equation*}
Likewise, we have
\begin{equation*}\alpha^*_N-\alpha=\frac{d}{dv}g(v)\mathcal{O}\left(\frac{M}{N}\right)-g(v)o\left(N^{-1}\right).\end{equation*}
~
\end{proof}

\begin{proof}[Proof of Theorem \ref{thm:VaR convergency rate}]

By Condition (A5), $U_{t_1}^{[M]}$ is continuous. Therefore,
$$\inf_{x\in \rm I\!R}\left\{\Pr\{U_0-U_{t_1}^{[M]}<-x\}<\alpha \right\}=\left\{ x\in \rm I\!R; \Pr\{U_0-U_{t_1}^{[M]}<-x\}=\alpha \right\}. $$

Similar to the proof of (28) in \cite{Gordy_Juneja-2010-MS}, we apply Taylor expansion to $\Pr\{U_{t_1}^{[M,N]}>v_1 \}$ in the following equation,
\begin{align*}
\alpha &=\Pr\{U_{t_1}^{[M,N]}>v \}-(v_1-v)g_N(v)-\frac{1}{2}(v_1-v)^2g^\prime_N(\tilde{v})+\mathcal{O}\left(N^{-1}\right),
\end{align*}
where $\title{v}$ is an appropriate value between $v$ and $v_1$.

By Condition (A5), $g_N^\prime(u)$ is uniformly bounded for all $v$. By Theorem \ref{lemma:VaR convergence rate},

\begin{eqnarray*}
   &   & \Pr\{U_{t_1}^{[M,N]}>v \}\\
   & = & \Pr\{U_{t_1}^{[M]}>v \}-g(v)\mathcal{O}\left(\sqrt{\frac{s_0\log M}{N\phi_0^2}}\right)+\frac{d}{dv}g(v)\mathcal{O}\left(\frac{s_0\log M}{N\phi_0^2}\right)+\mathcal{O}\left(\frac{1}{N^{3/2}}\right) \\
   & = & \alpha-g(v)\mathcal{O}\left(\sqrt{\frac{s_0\log M}{N\phi_0^2}}\right)-\frac{d}{dv}g(v)\mathcal{O}\left(\frac{s_0\log M}{N\phi_0^2}\right)+\mathcal{O}\left(\frac{1}{N^{3/2}}\right).
\end{eqnarray*}

Therefore, we have
\begin{align*}
v_1-v&=\frac{1}{g_N(v)}\left[g(v)\mathcal{O}\left(\sqrt{\frac{s_0\log M}{N\phi_0^2}}\right)-\frac{d}{dv}g(v)\mathcal{O}\left(\frac{s_0\log M}{N\phi_0^2}\right)+o\left(N^{-1}\right) \right].
\end{align*}

To derive the relation between $g_N(v_0)$ and $g(v)$, we observe that
\begin{align*}
g_N(v)-g(v)&=\int_{\rm I\!R}\big (g_N(v-\frac{w}{\sqrt{N}},w)-g_N(v,w) \big )dw \\
  &= \int_{\rm I\!R} -\frac{1}{\sqrt{N}}wg_N(\tilde{v},w)dw \\
  &=g(\tilde{v})\mathcal{O}\left(\sqrt{\frac{s_0\log M}{N\phi_0^2}}\right) ,
\end{align*}
where $\tilde{v}$ lies between $v-w/\sqrt{N}$ and $v$.

\begin{eqnarray*}
v_1-v &=& \left[\frac{1}{g(v)}-g(\tilde{v})\frac{g^\prime(v)}{g^2(v)} \mathcal{O}\left(\sqrt{\frac{s_0\log M}{N\phi_0^2}}\right)\right] \\
      & & \times \left[-g(v)\mathcal{O}\left(\sqrt{\frac{s_0\log M}{N\phi_0^2}}\right)-\frac{d}{dv}g(v)\mathcal{O}\left(\frac{s_0\log M}{N\phi_0^2}\right)+o\left(N^{-1}\right)\right] \\
      &=&\mathcal{O}\left(\sqrt{\frac{s_0\log M}{N\phi_0^2}}\right)+\left [\frac{g(\tilde{v})g'(v)}{g(v)}-\frac{1}{g(v)}\frac{d}{dv}g(v)  \right ]\mathcal{O}\left(\frac{s_0\log M}{N\phi_0^2}\right)+o\left(N^{-1}\right) .
\end{eqnarray*}

Likewise, we can prove that
$$v_1^*-v=\frac{1}{g(v)}\frac{d}{dv}g(v)\mathcal{O}\left(\frac{M}{N}\right)+o\left(N^{-1}\right).$$

If $N=o\left(\frac{M^2\phi_0^2}{s_0\log M}+2M+\frac{s_0\log M}{\phi_0^2}\right)$,
\begin{eqnarray*}&&\left|\mathcal{O}\left(\sqrt{\frac{s_0\log M}{N\phi_0^2}}\right)+\left [\frac{g(\tilde{v})g'(v)}{g(v)}-\frac{1}{g(v)}\frac{d}{dv}g(v)  \right ]\mathcal{O}\left(\frac{s_0\log M}{N\phi_0^2}\right)\right|\\&=&o\left |\frac{1}{g(v)}\frac{d}{dv}g(v)\mathcal{O}\left(\frac{M}{N}\right)\right |.\end{eqnarray*}
The desired result thus follows.
\end{proof}

\section*{Appendix B: Details of Numerical Studies} \label{sst:numerics}
This section contains the details for the numerical studies discussed in Section \ref{st:numerical study} including the data, the underlying models and their calibrated parameters.

\subsection*{B1. Settings for Rainbow Options in Section \ref{sst:rainbow option}}

To derive a benchmark utilizing existing closed form solution for pricing, we assume the underlying stock prices follow Black-Scholes Model, where the risk-free rate $r$, the volatility of each underlying stocks and the correlation between different underlying stocks remain constant from $T_0$ to $T$. Define $\rho_{ij}$ as the correlation between the $i$th and $j$th underlying stock and $\sigma_{ij}=\sigma_i^2+\sigma_j^2-2\rho_{ij}\sigma_i \sigma_j$ as the covariance. Define
\begin{eqnarray*}
\rho_{iij}&=&\frac{\sigma_i-\rho_{ij}\sigma_j}{\sigma_{ij}} \\
\rho_{ijk}&=&\frac{\sigma_i^2-\rho_{ij}\sigma_i \sigma_j-\rho_{ik}\sigma_i \sigma_k+\rho_{jk}\sigma_j \sigma_k}{\sigma_{ij}\sigma_{ik}} \\
d^\prime_2(S_0,K,\sigma^2)&=&\frac{\log(Ke^{-rT}/S_0)-\sigma^2T/2}{\sigma \sqrt{T}} \\
d^\prime_1(S_0,K,\sigma^2)&=&d^\prime_2+\sigma \sqrt{T},
\end{eqnarray*}
where $i,j,k=1,...,10$. Similar to the closed form solution of the ``call on min" rainbow option given in \cite{Johnson-1987-JFQA}, the option price at any given time $t \in [0,T]$ can be written as

\begin{eqnarray*}
c_t & = & 100\frac{S_{1t}}{S_{1,0}} N_n \bigg(d_1(\frac{S_{1t}}{S_{1,0}},K,\sigma_1^2),-d^\prime_1(\frac{S_{1t}}{S_{1,0}},\frac{S_{2t}}{S_{2,0}},\sigma_{12}^2),...,-d^\prime_1(\frac{S_{1t}}{S_{1,0}},\frac{S_{nt}}{S_{n,0}},\sigma^2_{1n}), \\
    & & \enspace -\rho_{112},-\rho_{113},...,\rho_{123},... \bigg)\\
    & + & 100\frac{S_{2t}}{S_{2,0}} N_n\bigg(d_1(\frac{S_{2t}}{S_{2,0}},K,\sigma^2_2),-d^\prime_1(\frac{S_{2t}}{S_{2,0}},\frac{S_{1t}}{S_{1,0}},\sigma^2_{12},...,-d^\prime_1(\frac{S_{2t}}{S_{2,0}},\frac{S_{nt}}{S_{n,0}},\sigma^2_{2n}), \\
    & & \enspace -\rho_{212}, -\rho_{223},...,\rho_{213},...  \bigg) \\
    & + & \cdots \\
    & + & 100\frac{S_{nt}}{S_{n,0}}N_n\bigg(d_1(\frac{S_{nt}}{S_{n,0}},K,\sigma^2_n),-d^\prime_1(\frac{S_{nt}}{S_{n,0}},\frac{S_{1t}}{S_{1,0}},\sigma^2_{1n}),...,-d^\prime_1(\frac{S_{nt}}{S_{n,0}},\frac{S_{n-1,t}}{S_{n-1,0}},\sigma^2_{n-1,n}), \\
    & & \enspace -\rho_{n1n},-\rho_{n2n},...,\rho_{n12},...)\bigg)\\
    & - & 100Ke^{-r(T-t)}N_n\bigg(d_2(\frac{S_{1t}}{S_{1,0}},K,\sigma^2_1),d_2(\frac{S_{2t}}{S_{2,0}},K,\sigma^2_2),...,d_2(\frac{S_{nt}}{S_{n,0}},K,\sigma^2_n),\rho_{12},\rho_{12},... \bigg) ,
\end{eqnarray*}
where $n=10$, $N_n(\cdot)$ is the cumulative distribution function of the $n$-dimensional standard normal distribution.

The option price at $T_0$ can thus be reduced to
\begin{eqnarray*}
c_0 & = & 100N_n \bigg(d_1(1,K,\sigma_1^2),-d^\prime_1(1,1,\sigma_{12}^2),...,-d^\prime_1(1,1,\sigma^2_{1n}),-\rho_{112},-\rho_{113},...,\rho_{123},... \bigg)\\
    & + & 100N_n\bigg(d_1(1,K,\sigma^2_2),-d^\prime_1(1,1,\sigma^2_{12},...,-d^\prime_1(1,1,\sigma^2_{2n}),-\rho_{212},-\rho_{223},...,\rho_{213},...  \bigg) \\
    & + & \cdots \\
    & + & 100N_n\bigg(d_1(1,K,\sigma^2_n),-d^\prime_1(1,1,\sigma^2_{1n}),...,-d^\prime_1(1,1,\sigma^2_{n-1,n}),-\rho_{n1n},-\rho_{n2n},...,\rho_{n12},...) \bigg)\\
    & - & 100Ke^{-rT}N_n\bigg(d_2(1,K,\sigma^2_1),d_2(1,K,\sigma^2_2),...,d_2(1,K,\sigma^2_n),\rho_{12},\rho_{12},... \bigg).
\end{eqnarray*}

Parameters in the dynamics of the underlying stocks include the risk-free rate $r$, the volatility $\sigma_i$, the drift $\mu_i$, the current price $S_{i0}$ and the correlation between different stocks $\rho_{ij}$, where $i,j=1,...,10$. They are reasonably chosen based on the observation of commonly traded stocks in the market. $500$ daily historical underlying stock prices are simulated assuming Black-Scholes as the underlying model. We set the volatility $\sigma_1$, $\sigma_2$ and the correlation between $S_1$, $S_2$ relatively large so that $S_1$ and $S_2$ represent significant variables in the regression. 

The starting historical price $S_{i,-1}$, the daily drift $\mu_i$ and the volatility $\sigma_i$ are shown in Table \ref{table:rainbow option parameters1} while the correlation matrix is presented in Table \ref{table:rainbow option parameters2}. The numerical results can be found in Section \ref{sst:rainbow option}.

\begin{table}[]
\centering
\caption{Parameters in the Underlying Model}
\label{table:rainbow option parameters1}
\begin{tabular}{lll}
\hline
$S_{i,-1}$        & $\mu_i$          & $\sigma_i$        \\ \hline
80.38723 & 1.1015E-06 & 0.0085263 \\
42.70244 & 1.5939E-06 & 0.0093093 \\
67.57745 & 3.4755E-06 & 0.0024763 \\
85.70454 & 3.8621E-05 & 0.0021646 \\
58.11831 & 8.6745E-05 & 0.0042942 \\
32.29635 & 7.4338E-05 & 0.0025601 \\
57.28909 & 9.0098E-05 & 0.0044424 \\
68.65604 & 1.1443E-05 & 0.0010326 \\
86.43502 & 7.7736E-05 & 0.0016128 \\
81.60649 & 1.2489E-05 & 0.0013172 \\ \hline
\end{tabular}
\end{table}

\begin{table}[]\footnotesize
\centering
\caption{Parameters in the Correlation Matrix}
\label{table:rainbow option parameters2}
\begin{tabular}{lllllllllll}
\hline
  & $S_1$                & $S_2$                & $S_3$ & $S_4$                & $S_5$                & $S_6$                & $S_7$                & $S_8$                & $S_9$                & $S_{10}$               \\  \hline
$S_1$  & 1.00000 & 0.55000 & 0.29311 & 0.28272 & 0.23681 & 0.33050 & 0.34773 & 0.39159 & 0.29665 & 0.23986 \\
$S_2$  & 0.55000 & 1.00000 & 0.28613 & 0.27540 & 0.37854 & 0.38001 & 0.25678 & 0.32052 & 0.26683 & 0.28365 \\
$S_3$  & 0.29311 & 0.25510 & 1.00000 & 0.31191 & 0.39619 & 0.32266 & 0.27440 & 0.26772 & 0.39976 & 0.28598 \\
$S_4$  & 0.28613 & 0.33050 & 0.27440 & 1.00000 & 0.25510 & 0.23745 & 0.22811 & 0.25273 & 0.22504 & 0.35783 \\
$S_5$  & 0.28273 & 0.38001 & 0.22811 & 0.25273 & 1.00000 & 0.24183 & 0.25727 & 0.29702 & 0.30817 & 0.33151 \\
$S_6$  & 0.27540 & 0.32266 & 0.25727 & 0.29702 & 0.39976 & 1.00000 & 0.25681 & 0.21482 & 0.32993 & 0.20017 \\
$S_7$  & 0.31191 & 0.23745 & 0.25681 & 0.21482 & 0.22504 & 0.21862 & 1.00000 & 0.28263 & 0.29389 & 0.24210 \\
$S_8$  & 0.23681 & 0.24183 & 0.39159 & 0.28263 & 0.30817 & 0.23986 & 0.35783 & 1.00000 & 0.21862 & 0.23128 \\
$S_9$  & 0.37854 & 0.34773 & 0.32052 & 0.29665 & 0.32993 & 0.28365 & 0.33151 & 0.24210 & 1.00000 & 0.37021 \\
$S_{10}$ & 0.39619  & 0.25678 & 0.26772 & 0.26683 & 0.29389 & 0.28598 & 0.20017 & 0.23128 & 0.37021 & 1.00000   \\
\hline
\end{tabular}
\end{table}

\subsection*{B2. Settings for Rainbow Swptions in Sections \ref{sst:european swaption} and \ref{sst:Bermudan swaption}}
Our formulation follows \cite{Brigo_Mercurio-2007} [Section 6.3.1] that assumes lognormal distribution of forward rates. The dynamics of forward rates $L_i(t)$ under $\mathbb{Q}^i$ are, respectively,
\begin{eqnarray*}
\label{eq:dynamics of forward rates}
& & i<j,\quad t\leq T_i: \qquad dL_j(t)=\sigma _j(t)L_j(t)\sum _{k=i+1}^{j}\frac{\rho _{kj}\delta _k \sigma _k(t)L_k(t)}{1+\delta _kL_k(t)}dt+\sigma _j(t)L_j(t)dZ_j(t)\\
& & i=j,\quad t\leq T_{i-1}: \quad dL_j(t)=\sigma _j(t)L_j(t)dZ_j(t)\\
& & i>j,\quad t\leq T_{j-1}: \quad dL_j(t)=-\sigma _j(t)L_j(t)\sum _{k=i+1}^{j}\frac{\rho _{kj}\delta _k \sigma _k(t)L_k(t)}{1+\delta _kL_k(t)}dt+\sigma _j(t)L_j(t)dZ_j(t),\\
\end{eqnarray*}
where $Z$ is a Brownian motion under measure $\mathbb{Q}^i$, $Z_i$, $Z_j$ are Brownian motions of different forward rates $L_i(t)$ whose instantaneous correlation with $L_j(t)$ is $\rho=(\rho_{ij})_{i,j=1,2,...}$. The measure associated with zero-coupon bonds maturing at time $T_i$ is denoted by $\mathbb{Q}^i$. Note that all equations in equation \eqref {eq:dynamics of forward rates} admit a unique strong solution if $\sigma_j(\cdot )$ are bounded.

In order to fully specify the forward rates dynamics in the LFM, instantaneous volatilities and correlation function have to be determined. A time-homogenous function to parameterize instantaneous volatilities and correlation is widely adopted. The term ``time-homogenous" here indicates that the function is time-dependent, and the time dependency is tied to the time left to reach maturity of the underlying swap. In our example, we apply one of the most commonly used parametric forms, namely
\begin{eqnarray*}
\label{eq:volatility function}
\sigma _i(t) &= & \psi _i\nu (T_{i-1}-t,\gamma ) \\
             & = & \psi _i[\{(T_{i-1}-t)\gamma _1+\gamma _2\}e^{-(T_{i-1}-t)\gamma _3}+\gamma _4\rbrack,
\end{eqnarray*}
where $\gamma=(\gamma _1,\gamma _2,\gamma _3,\gamma _4)$ is a parameter set, $\psi _i$ is a correction parameter that fits the volatilities more closely to market data. This function has a ``humped'' shape which can be interpreted descriptively with economic knowledge. 

For instantaneous correlation $\rho$, its parameterized form suggested in 
\cite{M_Joshi-2003} and 
\cite{R_Rebonato-2002} is given by
\begin{equation}
\label{eq:correlation function}	
\rho _{ij}=e^{-\beta |i-j|}.
\end{equation}

To calibrate parameters in instantaneous volatility and correlation, we take the market data as input
$$L_0=[L(T_0,T_0,T_1),L(T_0,T_1,T_2),...,L(T_0,T_{19},T_{20})]$$
of initial annual forward rates and the annual ATM caplet volatility
$$\sigma^{\text{caplet}}=[\sigma^{\text{caplet}}_1,...,\sigma^{\text{caplet}}_{20}]$$
where $\sigma^{\text{caplet}}_i$ stands for the volatility of annual caplet resetting at $i$-th year and paying at $(i+1)$-th year. $i=1,2,\cdots,20$.

A recursive calibration algorithm starts by initializing $\gamma_0$, $\beta_0$ by appropriate guess. With $\gamma_0$, $\beta_0$, we can estimate $\widetilde{\psi}_i$ for $i=1,...,20$ so as to match the market volatility of the co-terminal caplets by,
\begin{eqnarray*}
\sigma_i^{\text{caplet}} &= & \frac{1}{T_{i-1}}\widetilde{\psi}_i^2\int_0^\infty \nu^2(T_{i-1}-s,\gamma_0)ds \\
 & = & \frac{1}{T_{i-1}}\widetilde{\psi}_i^2\big (\sum_{s=0}^{i-1}((T_{i-1}-T_s)\gamma_{1,0}+\gamma_{2,0})e^{-(T_{i-1}-T_s)\gamma_{3,0}}+\gamma_{4,0}\big ).
\end{eqnarray*}

Given those $\widetilde{\psi}_i$'s, re-estimate $\gamma$, $\beta$ by
\begin{align} \label{calibration of volatility}
\argmin_{\gamma,\beta}|\sigma_i-\hat{\sigma}_i(\beta,\gamma;\widetilde{\psi})|^2,
\end{align}
where $\sigma_i$ are Black volatility for $i$ NC $\alpha$ swaptions, $\hat{\sigma}_i(\beta,\gamma;\widetilde{\psi})$ is the model volatility adopted in 
\cite{R_Rebonato-2002}. The corresponding formula approximates the lognormal forward LIBOR model swaption volatility by
$$\hat{\sigma}^2_i(0)=\sum_{j,k=3}^i \frac{w_j(0)w_k(0)L_j(0)L_k(0)\rho_{jk}}{S_{2,i}^2(0)}\sum_{s=0}^2 \sigma_j(T_s)\sigma_k(T_s),\enspace \enspace i\geq 3,$$
where $w_j(0)=\frac{\delta_jL(0,T_2,T_j)}{\sum_{k=3}^i\delta_kL(0,T_2,T_k)}$ and $S_{2,i}(0)$ is the ATM swap rate for $i$ NC $2$ swaptions. Substitute functional forms in formula \eqref{eq:volatility function} and \eqref{eq:correlation function} for instantaneous volatility and correlation, $\hat{\sigma}_i^2(0)$ can be expressed as a function of parameter $\gamma,\beta,\psi$.

Re-estimating $\gamma,\beta$ can be achieved by solving the minimization problem in formula \eqref{calibration of volatility} after which, re-estimate $\psi$ iteratively is carried out. The iteration procedure stops when either convergence or the maximum number of iteration is reached.

We put a constraint on the calibration of $\psi$ such that $1-0.1\leq\psi_i\leq 1+0.1$ for all $i$. This constraint requires all $\psi_i$ to be close to one so that the term structure's qualitative behavior could be captured in time. The functional form of instantaneous volatility and correlation are constructed to produce a smooth shape for the term structure of volatility at all instants, since the typical erratic behavior of piecewise-constant assumption can be improved by linear/exponential functions. Numerical results are shown in Section \ref{sst:european swaption}.

\bibliographystyle{asa}
\bibliography{VaRHD}
\end{document}